\documentclass[12pt]{article}
\usepackage{amsmath, amssymb, amsthm,authblk}
\usepackage[authoryear]{natbib}
 \usepackage{color}
 \usepackage{algorithm,algorithmic}
 \usepackage{longtable}
 \usepackage{ulem}

\usepackage{geometry, graphicx,float,appendix,lmodern}
\usepackage[colorlinks,citecolor=blue,urlcolor=blue]{hyperref}
\usepackage{multirow}

\usepackage{inputenc}
\usepackage[english]{babel}
\usepackage{amsthm}
\usepackage{xcolor}
\usepackage{courier}
\usepackage{amssymb}
\usepackage{mathrsfs}
\usepackage{caption}
\usepackage{subcaption}
\usepackage{enumitem}
\usepackage{comment}
\usepackage{multirow,setspace}
\renewcommand{\baselinestretch}{1.5}
\makeatletter
\newcommand{\cM}{{\cal M}}

\newcommand{\cD}{\mbox{$\mathcal{D}$}}

\def\spacingset#1{\renewcommand{\baselinestretch}%
{#1}\small\normalsize} \spacingset{1}

\newtheorem{theorem}{Theorem}
\newtheorem{lemma}{Lemma}
\newtheorem{corollary}{Corollary}

\newtheorem{remark}{Remark}

\renewcommand{\hat}{\widehat}

\newcommand{\argmax}{{\rm argmax}}

\newcommand{\tr}{\mbox{tr}}

\newcommand{\bA}{{\mathbf A}}
\newcommand{\bB}{{\mathbf B}}

\newcommand{\bG}{{\mathbf G}}
\newcommand{\bH}{{\mathbf H}}
\newcommand{\bI}{{\mathbf I}}

\newcommand{\bM}{{\mathbf M}}

\newcommand{\bQ}{{\mathbf Q}}
\newcommand{\bP}{{\mathbf P}}

\newcommand{\bS}{{\mathbf S}}
\newcommand{\bU}{{\mathbf U}}
\newcommand{\bV}{{\mathbf V}}

\newcommand{\bZ}{{\mathbf Z}}
\newcommand{\ba}{{\mathbf a}}

\newcommand{\bff}{{\mathbf f}}

\newcommand{\bq}{{\mathbf q}}

\newcommand{\bs}{{\mathbf s}}
\newcommand{\bu}{{\mathbf u}}
\newcommand{\bv}{{\mathbf v}}

\newcommand{\bx}{{\mathbf x}}

\newcommand{\bz}{{\mathbf z}}

\newcommand{\bdelta} {\boldsymbol{\delta}}

\newcommand{\bSigma}{\boldsymbol{\Sigma}}

\newcommand{\bve}{\mbox{\boldmath$\varepsilon$}}

\newcommand{\bzero}{{\mathbf 0}}

\newcommand{\calD}{{\mathcal D}}

\newcommand{\calF}{{\mathcal F}}

\newcommand{\calM}{{\mathcal M}}

\def\6bullets{\bullet\bullet\bullet\bullet\bullet\bullet}

\bibpunct{(}{)}{;}{a}{,}{,}

\title{\bf Sparse-Group Factor Analysis for High-Dimensional Time Series}

\date{}

\begin{document}

 \author{Xin Wang\\
    Department of Mathematics and Statistics, San Diego State University\\
    and \\
    Xialu Liu \\
    Department of Management Information Systems, San Diego State University}

\maketitle

\begin{abstract}
Factor analysis is a widely used technique for dimension reduction in high-dimensional data. However, a key challenge in factor models lies in the interpretability of the latent factors. One intuitive way to interpret these factors is through their associated loadings. \cite{liu::wang2025} proposed a novel framework that redefines factor models with sparse loadings to enhance interpretability. In many high-dimensional time series applications, variables exhibit natural group structures. Building on this idea, our paper incorporates domain knowledge and prior information by modeling both individual sparsity and group sparsity in the loading matrix. This dual-sparsity framework further improves the interpretability of the estimated factors. We develop an algorithm to estimate both the loading matrix and the common component, and we establish the asymptotic properties of the resulting estimators. Simulation studies demonstrate the strong performance of the proposed method, and a real-data application illustrates how incorporating prior knowledge leads to more interpretable results.
\end{abstract}

\spacingset{1.9} 


\section{Introduction}
High-dimensional time series data are widely observed in various fields, such as economics \citep{stock1998,stock2002a,stock2002b,ma2018,chen2021dynamic}, finance \citep{lam2011,chang2015,massacci2017,wang2019, chen2019constrained}, environmental \citep{pan2008,lam2012} and medical sciences \citep{lindquist2008statistical,smith2014group}. When analyzing high-dimensional data, various challenges arise, referred to as the curse of dimensionality. For example, to achieve the same level of estimation accuracy as in lower dimensions, the required sample size must also increase exponentially for high-dimensional data analysis. A widely used approach to break the curse of dimensionality is factor analysis \citep{forni2000,bai2002,bai2003,forni2005,lam2011,lam2012,fan2022learning,chen2023statistical}, which assumes that high-dimensional data can be represented by a much lower-dimensional process, called factors. One issue that hinders the wide application of factor models is that factors are hard to interpret because the factors and the loading matrix are unobserved. To address this, \cite{liu::wang2025} developed an algorithm to obtain a sparse estimate of the loading matrix with the orthogonality constraint, facilitating the interpretation of factors. In this work, we incorporate prior knowledge to the sparsity assumption and propose a sparse-group estimator that aligns with certain domain theories, further enhancing the model interpretation.  

When analyzing high-dimensional data from applied disciplines, relevant prior information about common factors is often available. For example, macro-economists and financial economists model the yield rates at different maturities using three latent factors --- level, slope and curvatures \citep{diebold2006, diebold2006macroeconomy}. In business and finance, variables are often naturally grouped. For instance, financial researchers break stocks into different groups by size and book-to-market to mimic the underlying factors in returns for asset pricing \citep{fama1993common,feng2020taming}. \cite{wang2019} analyzed the financial data of 200 companies and the companies were grouped by industry. \cite{chen2022trade} analyzed the trading volumes between 24 countries. These countries are clustered based on their geographical locations (Europe, Asia, etc.) and economic status (developing or developed countries). In many cases, factors make different impacts on different groups. \cite{liu::wang2025} studied the tourism data in Hawaii and found that two groups drive the number of domestic tourists: people to escape the cold (factor 1) and to enjoy the beach and water activities (factor 2). They also discovered that visitors from high-latitude states load heavily on factor 1 and visitors from inland or low-latitude states load heavily on factor 2. In this work, we propose a sparse group factor model that accommodates both group level and within group level sparsity. This structure enables the incorporation of prior information and produces more interpretable factors. 

Several studies in the literature made efforts on incorporation of prior knowledge into factor models \citep{tsai2010constrained,chen2019constrained}. When utilizing group information, these methods impose pre-specified group-wise constraints on the factor models and only between-group sparsity is considered. In contras, the model we propose allows both individual zero loadings (within-group sparsity) and group zero loadings (between-group sparsity). Within-group or between-group sparsity is determined by a data-driven approach. An existing method that enforces group-level sparsity on loadings is sparse group principal component analysis \citep{guo2010principal, jenatton2010structured,lee2025sparse}. In these approaches, penalty functions \citep{huang2012selective} are employed to identify the sparse loading structures. As is known, the loading matrix is not uniquely defined, and can rotate in the loading space. Compared with these variations of principal component analysis that assume the columns of loading matrix are orthogonal, the proposed method relaxes this assumption and explores the entire loading space to identify the most sparse loading matrix, thus producing a more sparse estimate for the loading matrix and making the factors more interpretable.

The paper makes the following contributions: (1) The factor model we propose can accommodate both individual level sparsity and group level sparsity in the loading matrix. It utilizes the prior information and enhances the model interpretation. (2) Compared with sparse-group PCA, we follow the approach of \cite{liu::wang2025} to define the loading matrix with the most sparse structure while allowing the columns of the loading matrix to be non-orthogonal. This results in a more parsimonious model and provides a clearer interpretation of the latent factors. 

The rest of the paper is organized as follows. Section \ref{sec::model} introduces the sparse group factor model with group structure. Section \ref{sec:Est} presents the estimation methods for the proposed sparse group factor model. Section \ref{sec:them} investigates the theoretical properties of our proposed estimators. The simulation and real data analysis are presented in Section \ref{sec:sim} and Section \ref{sec:example}. Finally, Section \ref{sec:conclusion} concludes the paper. 

\section{Model}\label{sec::model}

We introduce some notations first. For a vector $\bz$, we use $z_i$ to denote its $i$-th element. For a $p_1\times p_2$ matrix $\bZ$, its ($i,j$)-th element is denoted by $z_{ij}$ and its $i$-th column is denoted by $\bz_i$. Furthermore, we use ${\cal M}(\bZ)$ to denote the space spanned by the columns of $\bZ$. Let $\Vert \bZ\Vert_F$ be the Frobenius norm of $\bZ$, where $\Vert \bZ\Vert_F = \sqrt{\sum_{i=1}^{p_1}\sum_{j=1}^{p_2}z_{ij}^2}$, $\Vert \bZ\Vert_2$ be the $L$-2 norm of $\bZ$, where $\Vert \bZ\Vert_2 = \left[\lambda_{\max}(\bZ^\top \bZ)\right]^{1/2} $ and $\lambda_{\max}(\cdot)$ is the maximum eigenvalue of a square matrix, and $\|\bZ\|_{\min}$ is the nonzero minimum singular value of $\bZ$. We also define $L_1$, $L_{\infty}$ and max norm as follows: $\Vert{\bf Z}\Vert_{1}  =\max_{1\leq j\leq p_2}\sum_{i=1}^{p_1}\vert z_{ij}\vert$, $\Vert{\bf Z}\Vert_{\infty}  =\max_{1\leq i\leq p_1}\sum_{j=1}^{p_2}\vert z_{ij}\vert$ and $\Vert{\bf Z}\Vert_{\max} =\max_{ij}\vert z_{ij}\vert$. 
For a vector $\bz$, we let $\Vert{\bf z}\Vert_{\infty}=\max_{i}\vert z_{i}\vert$ and $\Vert{\bf Z}\Vert_{2\rightarrow\infty}=\sup_{\Vert\bx\Vert_{2}=1}\Vert{\bf Z}{\bf x}\Vert_{\infty}$, which is called two-to-infinity norm studied in \cite{cape2019two}. For any $\{a_n\}$ and $\{b_n\}$, ``$a_n\asymp b_n$" means $\lim_{n \to \infty} a_n/b_n = c$, where $c$ is a positive constant, and ``$a_n \gtrsim b_n$" means $a_n^{-1}b_n = o(1)$.

\label{sec:model}
Let $\bx_t$ be an observed $p \times 1$ time series
$t=1,\ldots, n$. The general form of a factor model for a $p$-dimensional time series is
\begin{equation}
\bx_t=\bA \bff_t +\bve_t, \label{eq::model}    
\end{equation}
where $\bx_t$ is the $p$-dimensional time series, $\bff_t=(f_{t1}, f_{t2}, \ldots, f_{tr})^\top$ is a set of unobserved
(latent) factor time series with dimension $r$ that is much smaller than $p$,
the matrix $\bA$ is the loading matrix of the common factors, and $\bve_t$ is a noise process.

One of the important characteristics of factor models is that both factors $\bff_t$ and loading matrix $\bA$ are unobserved. Therefore, the interpretation of latent factors can be done via $\bA$ and $\bx_t$, but is really challenging. Another feature of factor models is that factors and the loading matrix are not uniquely defined. Specifically, $(\bA, \bff_t)$ in (\ref{eq::model}) can be replaced by $(\bA\bV, \bV^{-1}\bff_t)$, where $\bV$ is an invertible $r\times r$ matrix. Fortunately, the column space spanned by $\bA$, denoted by ${\cal M}(\bA)$ and called the loading space, is unique. \cite{liu::wang2025} improves the model interpretability by defining the most sparse loading matrix in the loading matrix. Specifically, they rewrite the factor models in (\ref{eq::model}) as
\begin{equation}
 \bx_t=\bA^s \bff_t^s +\bve_t, \label{eq::model112}   
\end{equation}
where $\bA^s$ is one of the matrices with most zero elements in the loading space ${\cal M}(\bA)$, and satisfies: (1) $\|\ba_i^s\|_2=\|\ba_i\|_2$;  (2) Let $m_i$ be the number of nonzero elements in $\ba_i^s$ and $ m_1 < m_2 < \ldots < m_r$. The factor models with sparse loadings (\ref{eq::model112}) can be re-expressed with the standardized loading matrix as follows,
\begin{equation}
   \bx_t=\bQ \bz_t +\bve_t, \label{eq::model2} 
\end{equation}
where $\bq_1=\frac{\ba_1^s}{\|\ba_1^s\|_2}$ and $\bq_i=\frac{\ba_i^s}{\|\bP_i \ba_i^s\|_2}$, for $i=2,\ldots, r$, where $\bP_i =\bI- \bQ_{(i)} (\bQ_{(i)}^\top \bQ_{(i)})^{-1} \bQ_{(i)}^\top$ and $\bQ_{(i)}=(\bq_1, \ldots, \bq_{i-1})$. The norm of $\bq_i$ is bounded by imposing a constraint that the remainder has a norm of 1 if $\bq_i$ is projected on the space spanned by $\{\bq_1, \ldots, \bq_{i-1}\}$, for $i=2,\ldots, r$. Compared with traditional factor models (\ref{eq::model}), model (\ref{eq::model2}) has more sparse loadings and may provide a clearer interpretation of factors; see example 1 in \cite{liu::wang2025}. 

High-dimensional time series are often naturally grouped. For example, \cite{liu::wang2025}
analyzed the numbers of domestic visitors from 49 states to Hawaii. The 49 states can be grouped by regions (South, Midwest, etc.). In many cases, there are more than one way to group the variables. For instance, \cite{chen2022trade} analyzed the trading volumes between 24 countries. These countries were clustered based on both their geographical locations (Europe, Asia, etc.) and economic status (developing/developed countries). 


To incorporate such prior information, we introduce the following notation. Let $J_i$ be the number of groups in ${\bf q}_i$, $d_{ij^{\prime}}$ be the group size of the $j^{\prime}$-th group in the $i$-th factor, and ${\cal{G}}_{i(j^\prime)}$ be the index set of elements in the $j^\prime$-th group, for $j^{\prime} = 1,\dots, J_i$, and ${\cal{G}}_i= \{{\cal{G}}_{i(1)}, {\cal{G}}_{i(2)}, \ldots, {\cal{G}}_{i(J_i)} \}$ for $i = 1,\dots, r$. We denote $\bq_{i(j^\prime)}$ as a $d_{ij^{\prime}}\times 1$ vector representing the loadings of the $j^\prime$-th group for factor $i$. It is worth mentioning that the group structure of different factors can be the same or different.

\begin{remark}
The algorithm we propose in Section \ref{sec:Est} still works if $m_1 \leq m_2  \leq \ldots \leq m_r$. When the sparsity level of certain columns in $\bQ$ is the same, our algorithm can recover one of the most sparse loading matrices and estimate the loading space effectively; see more details in \cite{liu::wang2025}. For technical convenience, we assume that $m_1< m_2 < \ldots <m_r$ throughout the paper. 
\end{remark}

\begin{figure}[H]
    \centering
    \includegraphics[scale=0.8]{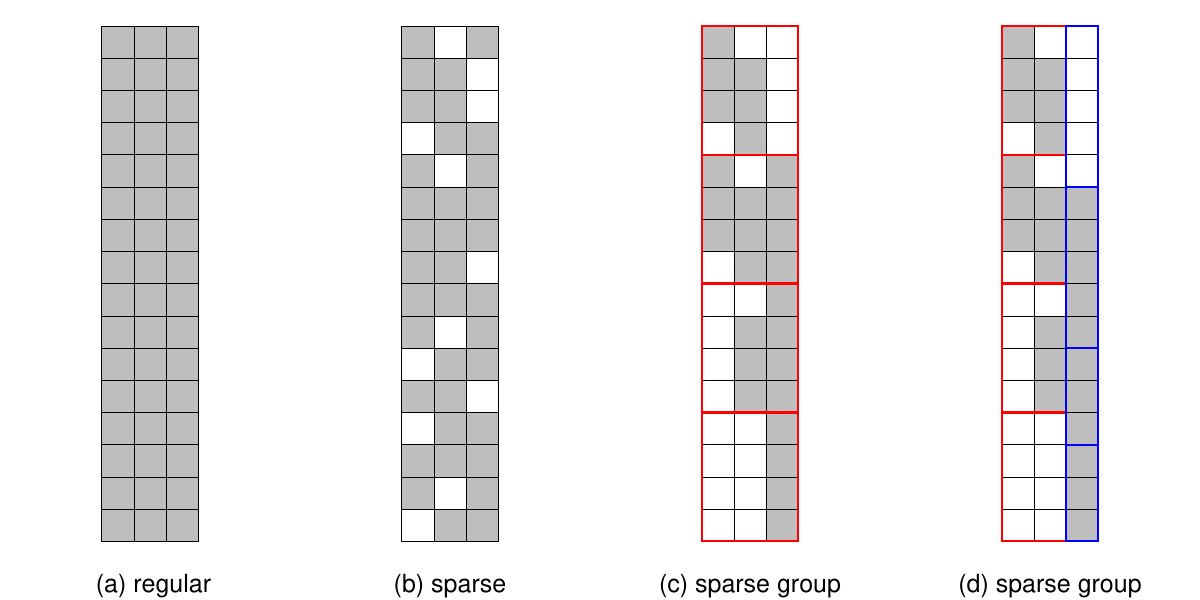}
    \caption{Examples of loading matrices}
    \label{fig:example}
\end{figure}

Figure \ref{fig:example} illustrates the loading matrices obtained by (a) the standard estimation method in \cite{lam2011}, (b) the sparse factor analysis method introduced in \cite{liu::wang2025}, and (c,d) the sparse-group factor analysis proposed in this paper. In these figures, gray cells represent nonzero loadings, and white cells represent zero loadings. Elements in different groups are separated by red or blue solid lines. In (c), there is only one way to group the time series with ${\cal G}_1 = {\cal G}_2 = {\cal G}_3$, where different groups are separated by red lines with total four groups. The group order is from the top to the bottom. In (c), we have $\bq_{1(3)}=\bq_{1(4)}=\bq_{2(4)}=\bq_{3(1)}=\mathbf{0}$. In (d) there are two ways to group the time series with  ${\cal G}_1={\cal G}_2 \neq {\cal G}_3$, where the first one is shown in red and the second one is shown in blue. Loadings for factor 1 and factor 2 have the same group structure, and loadings for factor 3 have a different group structure. In this example, group loadings $\bq_{1(3)}$, $\bq_{1(4)}$,  $\bq_{2(4)}$ and $\bq_{3(1)}$ are all zero. The group classification is assumed to be given, and for each factor which group structure will be used (red/blue) is also assumed to be known.


\begin{remark}
Variables from the same group are not necessarily adjacent as in the example shown in Figure \ref{fig:example}(d). However, for technical convenience, we assume that the loadings from the same group are arranged consecutively. In other words, $\bq_i$ can be written as $\bq_i = (\bq_{i(1)}^\top,\bq_{i(2)}^\top,\dots, \bq_{i(J_i)}^\top)^\top$, for $i=1,\ldots, r$.
\end{remark}

\section{Estimation}
\label{sec:Est}
In this section, we first briefly review the standard estimation method for the loading matrix proposed by \cite{lam2011} in Section \ref{subsec:reviewlam}, and then introduce our algorithm for obtaining a sparse group estimate in Section \ref{subsec:est_algorithm}.

\subsection{The standard estimation method}
\label{subsec:reviewlam}
Define
\[
\bSigma_x(h)=\frac{1}{n-h} \sum_{t=1}^{n-h} {\rm E} (\bx_t \bx_{t+h}^\top),\quad
\bM= \sum_{h=1}^{h_0} \bSigma_x(h) \bSigma_x(h)^\top,
\]
where $h_0$ is a pre-specified positive integer. Since $\{\bve_t\}$ has no serial dependence, we have
\begin{align}
\bM= \bA^s \left(\sum_{h=1}^{h_0} \bSigma_f^s(h) \bA^{s\top} \bA^s \bSigma_f^s(h)^\top \right) \bA^{s\top}, \label{eq::M}
\end{align}
where $\bSigma_f^s(h)=\sum_{t=1}^{n-h} {\rm E}(\bff_t^s \bff_{t+h}^{s\top})/(n-h)$. If the matrix in parentheses in (\ref{eq::M}) is full rank, the space spanned by the eigenvectors of $\bM$ corresponding to non-zero values is ${\calM}(\bA^s)$.

\cite{lam2011} defined the sample version of these matrices as follows
\[
\hat{\bSigma}_x(h)= \frac{1}{n-h} \sum_{t=1}^{n-h}\bx_t \bx_{t+h}^\top, \quad \hat{\bM}= \sum_{h=1}^{h_0} \hat{\bSigma}_x(h) \hat{\bSigma}_x(h)^\top.
\]

Thus, the loading space ${\cal M}(\bA^s)$ is estimated by ${{\cal M}(\hat{\bS})}$, where $\hat{\bS}=\{\hat{\bs}_1,\ldots, \hat{\bs}_r\}$ and $\hat{\bs}_i$ is the eigenvector of $\hat{\bM}$ corresponding to the $i$-th largest eigenvalue. In other words,
\begin{equation}
 \hat{\bS}= \argmax_{\bS^\top\bS=\bI_r} \text{tr}(\bS \hat{\bM} \bS^\top). \label{eq::obj}   
\end{equation}

\begin{remark} 
In practice, the number of factors, $r$, is unknown and must be estimated. A substantial body of research has addressed this problem; see, for example, \cite{bai2002,onatski2009,kapetanios2010,lam2012,han2022}. However, the primary focus of this paper is the estimation of factor loadings rather than the determination of $r$. Therefore, we assume that the number of factors is known throughout the paper.
\end{remark}


\subsection{Estimation with Sparse Group structure}
\label{subsec:est_algorithm}

In this section, we first introduce the optimization objective function for loading matrix estimation and then present the proposed algorithm. 

\subsubsection{The optimization problem}

To obtain the estimate we desire, the column space of this estimate should be close to ${\cal M}(\hat{\bS})$, and the number of zero elements needs to be as large as possible. Therefore, the objective function we would like to minimize is the distance between our estimate and $\cM(\hat{\bS})$ plus two penalty terms: one term is for the individual sparsity and the other term is for the group sparsity.

The distance of two linear spaces $\calM (\bU_1)$ and $\calM (\bU_2)$ with dimension of $r$ is defined as
\begin{align}
\label{eq:distance}
\calD (\calM(\bU_1), \calM(\bU_2))= \left(1- \frac{\tr(\bH_1 \bH_1^\top \bH_2 \bH_2^\top)}{r} \right)^{1/2},
\end{align}
where the columns of $\bH_i$ are an orthonormal basis of $\calM(\bU_i)$ for $i=1,2$ \citep{chang2015}. It is a quantity between 0 and 1. It is 1 if the two spaces are orthogonal and 0 if $\calM (\bU_1)=\calM (\bU_2)$. 

Penalty functions are widely used in regression models for obtaining sparse regression coefficients at the individual level or group level \citep{fan2010selective, huang2012selective}. Some studies also employed penalty functions to achieve sparsity at both individual and group levels. \cite{simon2013sparse} used sparse-group lasso in linear regression models. \cite{tugnait2022sparse} used sparse-group penalty for time series data. There are various penalty functions in the literature. The $L_1$ penalty (lasso) is the most popular convex penalty function \citep{tibshirani1996}. Two widely studied nonconvex penalties are the smoothly clipped absolute deviation (SCAD) penalty \citep{fan2001} and the minimax concave penalty (MCP) \citep{zhang2010}, both of them can achieve oracle properties. Usually, SCAD and MCP have similar performance. In this work, we adopt MCP, and the algorithm can be easily adapted to $L_1$ and SCAD. The MCP is defined as follows: $\mathcal{P}_{\gamma}(x,\lambda) = \lambda\vert x\vert-\frac{x^{2}}{2\gamma}$ if $\vert x\vert\leq\gamma\lambda$, and $\mathcal{P}_{\gamma}(x,\lambda) = \frac{1}{2}\gamma\lambda^{2}$ if $\vert x\vert>\gamma\lambda$, where $\gamma$ is fixed at 3 as in different literature \citep{breheny2011coordinate, breheny2015group}, and $\lambda$ is a tuning parameter selected based on data-driven criteria. To achieve the goal of identifying both individual level and group level sparsity, the penalty function will be applied on individual elements and group vectors in loadings, respectively. 


 Recall that $J_i$ denotes the number of groups in $\bq_i$, and $d_{ij^\prime}$ represents the group size of the $j^\prime$th group in $\bq_i$, for $j^\prime = 1,\dots, J_i$ and $i=1,\dots, r$.  Lemma D.2 in \cite{liu::wang2025} shows that minimizing the distance of ${\cD}({\cal M}(\bU_1), {\cal M} (\bU_2))$ is equivalent to minimizing $\sum_{i=1}^r \| \bH_1 \bH_1^\top- \mathbf{h}_{2i} \mathbf{h}_{2i}^\top\|_F^2$. Therefore, we can use the following steps to estimate the columns of $\bQ = (\bq_1,\bq_2,\dots, \bq_r)$ sequentially, where $q_{ij}$ is the $(i,j)$-th element in $\bQ $, and $\bq_{i(j^\prime)}$ is the $j^{\prime}$-th group in  $\bq_i$:
\begin{enumerate}
\item We estimate $\bq_1$ by solving:
\begin{align}
\label{eq:step1}
\hat{{\bf q}}_{1} & = \arg \min_{\bq_1}\frac{1}{2}\Vert\hat{{\bf S}}\hat{{\bf S}}^{\top}-{\bf q}_{1}{\bf q}_{1}^{\top}\Vert_{F}^{2}+\sum_{j=1}^{p}\mathcal{P}_{\gamma}\left(\vert q_{1j}\vert,\lambda_1\right)+\sum_{j^{\prime}=1}^{J_{1}}\mathcal{P}_{\gamma}\left(\Vert{\bf q}_{1\left(j^{\prime}\right)}\Vert_2,\sqrt{d_{1j^{\prime}}}\lambda_2\right)\\
 & \text{subject to }{\bf q}_{1}^{\top}{\bf q}_{1}=1. \nonumber
\end{align}

\item Let $\tilde{\bs}_1=\hat{\bq}_1$. For $i=2, \ldots, r$, we do the following
\begin{enumerate}
    \item Let $\tilde{\bS}_i= (\tilde{\bs}_1, \ldots, \tilde{\bs}_{i-1})$.
    \item  Estimate $\hat{\bq}_i$ by solving:
   \begin{align}
   \label{eq:step2}
\hat{{\bf q}}_{i} & =\arg \min_{\bq_i}\frac{1}{2}\Vert\hat{{\bf S}}\hat{{\bf S}}^{\top}-{\bf s}_{i}{\bf s}_{i}^{\top}\Vert_{F}^{2}+\sum_{j=1}^{p}\mathcal{P}_{\gamma}\left(\vert q_{ij}\vert,\lambda_1\right)+\sum_{j^{\prime}=1}^{J_{i}}\mathcal{P}_{\gamma}\left(\Vert{\bf q}_{i\left(j^{\prime}\right)}\Vert_2,\sqrt{d_{ij^{\prime}}}\lambda_2\right)\\
 & \text{subject to }{\bf s}_{i}=\left({\bf I}-\tilde{{\bf S}}_{i}\tilde{{\bf S}}_{i}^{\top}\right){\bf q}_{i}\text{ and }{\bf s}_{i}^{\top}{\bf s}_{i}=1. \nonumber
\end{align}

\item Set $\tilde{\bs}_i = (\bI- \tilde{\bS}_i \tilde{\bS}_i^\top)\hat{\bq}_i$.
\item Let $\hat{\bQ}= (\hat{\bq}_1, \ldots, \hat{\bq}_r)$.
\end{enumerate}

\end{enumerate}


In \eqref{eq:step1} and \eqref{eq:step2}, each objective function consists of three terms. The first term measures the distance between $\hat{\bS}$ and the proposed estimator. The second term imposes individual level sparsity in $\bq_i$, and the third term is for the group level sparsity in $\bq_i$. Note that we multiply $\sqrt{d_{ij'}}$ by the tuning parameter $\lambda_2$ in the third term in order to balance the penalty across groups of different sizes.


\subsubsection{The algorithm}
The optimization in \eqref{eq:step1} and \eqref{eq:step2} can be reformatted as the following general minimization problem:
\begin{align}
    \label{eq:optimizaiton0}
\hat{{\bf q}} & =\arg \min_{{\bf q}}\frac{1}{2}\Vert{\bf G}-{\bf B}{\bf q}{\bf q}^{\top}{\bf B}\Vert_{F}^{2}+\sum_{j=1}^{p}\mathcal{P}_{\gamma}\left(\vert q_{j}\vert,\lambda_1\right)+\sum_{j^{\prime}=1}^{J}\mathcal{P}_{\gamma}\left(\Vert{\bf q}_{\left(j^{\prime}\right)}\Vert_2,\sqrt{d_{j^{\prime}}}\lambda_2\right)\\
 & \text{subject to }{\bf q}^{\top}{\bf BB}{\bf q}=1. \nonumber
\end{align}   
In this general optimization problem in \eqref{eq:optimizaiton0}, we use $\bq$ to represent $\bq_i$ and $J$ represents $J_i$.  In particular, for the problem \eqref{eq:step1}, we set $\bG = \hat{\bS} \hat{\bS}^\top$, $\bB = \bI$; and for the problem in \eqref{eq:step2}, we set $\bB = \bI- \tilde{\bS}_i \tilde{\bS}_i^\top$, which satisfies $\bB\bB = \bB$ and $\bB^\top = \bB$. Furthermore, we have $\Vert{\bf G}-{\bf B}{\bf q}{\bf q}^{\top}{\bf B}\Vert_{F}^{2} = \tr(\bG\bG) - 2\bq^\top\bB\bG \bB \bq + 1$, so the first component in \eqref{eq:optimizaiton0} is equivalent to $-\bq^\top \bB \bG \bB \bq$ in the minimization problem.

To solve this optimization problem with constraints and penalty functions in \eqref{eq:optimizaiton0}, we use the alternating direction method of multipliers (ADMM) algorithm \citep{boyd2011distributed}, which is widely used in the literature \citep{ma2017concave, wang2023spatial, tugnait2022sparse}. To implement the ADMM algorithm, we first rewrite the optimization problem in \eqref{eq:optimizaiton0} as follows by introducing $\bs$ and $\bdelta$, which allow the complex problem to be decomposed into simpler sub-problems:
\begin{align}
    \label{eq:optimizaiton}
 & \min_{{\bf q},{\bf s},\bdelta}-{\bf s}^{\top}{\bf G}{\bf B}{\bf q}+\sum_{j=1}^{p}\mathcal{P}_{\gamma}\left(\vert q_{j}\vert,\lambda_1\right)+\sum_{j^{\prime}=1}^{J}\mathcal{P}_{\gamma}\left(\Vert \bdelta_{\left(j^{\prime}\right)}\Vert_2,\sqrt{d_{j^{\prime}}}\lambda_2\right)\\
 & \text{subject to }{\bf s}={\bf Bq},\,\bdelta={\bf q},\,{\bf s}^{\top}{\bf s}=1. \nonumber
\end{align} 

In the ADMM algorithm, the augmented Lagrangian corresponding to \eqref{eq:optimizaiton} has the following form,
\begin{align}
  L\left({\bf s},{\bf q},\bdelta,{\bf v}_{1},{\bf v}_{2}\right)= & -{\bf s}^{\top}{\bf G}{\bf B}{\bf q}+{\bf v}_{1}^{\top}\left({\bf s}-{\bf B}{\bf q}\right)+\frac{\rho_{1}}{2}\Vert{\bf s}-{\bf B}{\bf q}\Vert_{2}^{2}+\\
 & {\bf v}_{2}^{\top}\left(\bdelta-{\bf q}\right)+\frac{\rho_{2}}{2}\Vert\bdelta-{\bf q}\Vert_{2}^{2}+ \nonumber\\
 & +\sum_{j=1}^{p}\mathcal{P}_{\gamma}\left(\vert q_{j}\vert,\lambda_1\right)+\sum_{j^{\prime}=1}^{J}\mathcal{P}_{\gamma}\left(\Vert\bdelta_{\left(j^{\prime}\right)}\Vert_2,\sqrt{d_{j^{\prime}}}\lambda_2\right) \nonumber\\
 & \text{subject to }{\bf s}^{\top}{\bf s}=1,  \nonumber
\end{align}
where $\bv_1$ and $\bv_2$ are $p$-dimensional vectors containing the Lagrange multipliers, and $\rho_1$ and $\rho_2$ are fixed penalty parameters. Here, we set them at 1 as in \citet{ma2017concave} and \citet{wang2023spatial}.
Then, we can update $\bs, \bq, \bdelta, \bv_1, \bv_2$ iteratively. At the $(l+1)$-th iteration, given the current values of $\bs^{(l)}, \bq^{(l)}$, $\bdelta^{(l)}$, $\bv_1^{(l)}$ and $\bv_2^{(l)}$, the updates of $\bs, \bq, \bdelta, \bv_1, \bv_2$ are
\begin{align}
{\bf s}^{\left(l+1\right)} & =\arg\min_{{\bf s}^{\top}{\bf s}=1}L\left({\bf s}, {\bf q}^{\left(l\right)}, \bdelta^{\left(l\right)},{\bf v}_{1}^{\left(l\right)},{\bf v}_{2}^{\left(l\right)}\right),\label{eq:updates0}\\
{\bf q}^{\left(l+1\right)} & =\arg\min_{{\bf q}}L\left({\bf s}^{\left(l+1\right)}, {\bf q},\bdelta^{\left(l\right)},{\bf v}_{1}^{\left(l\right)},{\bf v}_{2}^{\left(l\right)}\right), \label{eq:updateq0}\\
\bdelta^{\left(l+1\right)} & =\arg\min_{\bdelta}L\left({\bf s}^{\left(l+1\right)},{\bf q}^{\left(l+1\right)}, \bdelta,{\bf v}_{1}^{\left(l\right)},{\bf v}_{2}^{\left(l\right)}\right), \label{eq:updatedelta0}\\
{\bf v}_{1}^{\left(l+1\right)} & ={\bf v}_{1}^{\left(l\right)}+\rho_{1}\left({\bf s}^{\left(l+1\right)}-{\bf B}{\bf q}^{\left(l+1\right)}\right), \label{eq:updatev1}\\
{\bf v}_{2}^{\left(l+1\right)} & ={\bf v}_{2}^{\left(l\right)}+\rho_{2}\left(\bdelta^{\left(l+1\right)}-{\bf q}^{\left(l+1\right)}\right). \label{eq:updatev2}
\end{align}

To update $\bs$, minimizing \eqref{eq:updates0} is equivalent to minimizing the following objective function with respect to $\bs$:
$
-{\bf s}^{\top}{\bf G}{\bf B}{\bf q}^{\left(l\right)}-\rho_{1}{\bf s}^{\top}{\bf B}{\bf q}^{\left(l\right)}+{\bf s}^{\top}{\bf v}_{1}^{\left(l\right)}=-{\bf s}^{\top}\left({\bf G}{\bf B}{\bf q}^{\left(l\right)}+\rho_{1}{\bf B}{\bf q}^{\left(l\right)}-{\bf v}_{1}^{\left(l\right)}\right),
$
subject to $\bs^\top \bs = 1$. By Cauchy-Schwarz inequality, the update of $\bs^{(l+1)}$ is
\begin{equation}
    \label{eq:update_s}
    {\bf s}^{\left(l+1\right)}=\frac{{\bf G}{\bf B}{\bf q}^{\left(l\right)}+\rho_{1}{\bf B}{\bf q}^{\left(l\right)}-{\bf v}_{1}^{\left(l\right)}}{\Vert{\bf G}{\bf B}{\bf q}^{\left(l\right)}+\rho_{1}{\bf B}{\bf q}^{\left(l\right)}-{\bf v}_{1}^{\left(l\right)}\Vert_{2}}.
\end{equation}

Note that $\bB \bB = \bB$, and $\bB = \bB^\top$. To update $\bq$ in \eqref{eq:updateq0}, it is equivalent to minimizing the following objective function with respect to $\bq$:
\begin{align*}
 & \frac{1}{2}{\bf q}^{\top}\left(\rho_{1}{\bf B}+\rho_{2}{\bf I}\right){\bf q}-{\bf q}^{\top}\left({\bf B}{\bf G}{\bf s}^{\left(l+1\right)}+\rho_{1}{\bf B}{\bf s}^{\left(l+1\right)}+{\bf B}{\bf v}_{1}^{\left(l\right)}+\rho_{2}\bdelta^{\left(l\right)}+{\bf v}_{2}^{\left(l\right)}\right)\\
 & +\sum_{j=1}^{p}\mathcal{P}_{\gamma}\left(\vert q_{j}\vert,\lambda_1\right).
\end{align*}
Let $\rho_{1}{\bf B}+\rho_{2}{\bf I}={\bf R}^{\top}{\bf R}$, where ${\bf R}$
is an upper triangular matrix with diagonal elements, and define ${\bf b}={\bf B}{\bf G}{\bf s}^{\left(l+1\right)}+\rho_{1}{\bf B}{\bf s}^{\left(l+1\right)}+{\bf B}{\bf v}_{1}^{\left(l\right)}+\rho_{2}\bdelta^{\left(l\right)}+{\bf v}_{2}^{\left(l\right)}$. Then we can re-rewrite the optimization problem as follows:
\begin{equation}
\label{eq:update_q}
\min_{\bq} \frac{1}{2}\Vert\left({\bf R}^{\top}\right)^{-1}{\bf b}-{\bf R}{\bf q}\Vert_{2}^{2}+\sum_{j=1}^{p}\mathcal{P}_{\gamma}\left(\vert q_{j}\vert,\lambda_1\right).   
\end{equation}
\eqref{eq:update_q} can be solved using a gradient algorithm with the MCP penalty. We use the R package {\it ncvreg} \citep{breheny2011coordinate} to obtain the solution for a fixed value of $\lambda$. Note that $L_1$ penalty or SCAD can also be used here. 

When updating $\bdelta$ group-wise, the update for each group $\bdelta_{(j^\prime)}$ is equivalent to solving the  following minimization problem with respect to $\bdelta_{(j^\prime)}$:
\[
\frac{\rho_{2}}{2}\Vert{\bdelta}_{\left(j^{\prime}\right)}-\left({\bf q}_{\left(j^{\prime}\right)}^{\left(l+1\right)}-\rho_{2}^{-1}{\bf v}_{2\left(j^{\prime}\right)}^{\left(l\right)}\right)\Vert_{2}^{2}+\mathcal{P}_{\gamma}\left(\Vert{\bdelta}_{\left(j^{\prime}\right)}\Vert_2,\sqrt{d_{j^{\prime}}}\lambda_{2}\right).
\]
Let
${\bf u}^{\left(l+1\right)}={\bf q}^{\left(l+1\right)}-\rho_{2}^{-1}{\bf v}_{2}^{\left(l\right)}$. For
the MCP, the update of $\bdelta_{(j^\prime)}$ is given by:
\begin{equation}
{\bdelta}_{\left(j^{\prime}\right)}^{\left(l+1\right)}=\begin{cases}
\frac{S\left({\bf u}_{\left(j^{\prime}\right)}^{\left(l+1\right)},\sqrt{d_{j^{\prime}}}\lambda_{2}/\rho_{2}\right)}{1-1/\left(\gamma_{2}\rho_{2}\right)} & \text{if \ensuremath{\Vert{\bf u}_{\left(j^{\prime}\right)}^{\left(l+1\right)}\Vert_2\leq\sqrt{d_{j^{\prime}}}\gamma_{2}\lambda_{2}}}\\
{\bf u}_{\left(j^{\prime}\right)}^{\left(l+1\right)} & \text{if }\Vert{\bf u}_{\left(j^{\prime}\right)}^{\left(l+1\right)}\Vert_2>\sqrt{d_{j^{\prime}}}\gamma_{2}\lambda_{2},
\end{cases}
\label{eq:update_delta}
\end{equation}
where $S({\bf x},\lambda)=(1-\lambda/\Vert{\bf x}\Vert_2)_{+}{\bf x},$
and $(x)_{+}=x$ if $x>0$, 0 otherwise.

In summary, the computational algorithm can be summarized as follows.
\begin{algorithm}[H]
	\caption*{{\bf{Algorithm}:} The optimization algorithm}
	\begin{algorithmic}[1]
     	\REQUIRE: Initialize $\bq^{(0)}$, $\bdelta^{(0)} = \bq^{(0)}$,  $\bv_1^{(0)} = \bf{0}$, and  $\bv_2^{(0)} = \bf{0}$. 
	    \FOR {$i=1$}
         \STATE Set $\bB = \bI$
         \FOR{$l=1,2,\dots,...$}
          \STATE Update $\bs_1$ by \eqref{eq:update_s},  $\bq_1$ by minimizing \eqref{eq:update_q}, $\bdelta_1$ by \eqref{eq:update_delta}, $\bv_1$ by \eqref{eq:updatev1} and $\bv_2$ by \eqref{eq:updatev2}.
         \STATE Stop and get $\tilde{\bs}_1$ and $\hat{\bq}_1$ if convergence criterion is met.
        \ENDFOR
        \ENDFOR
        \FOR {$i=2,\dots ,r$}
        \STATE{Compute $\tilde{\bS}_i = (\tilde{\bs}_1,\dots, \tilde{\bs}_{i-1})$ and $\bB = \bI - \tilde{\bS}_i\tilde{\bS}_i^\top.$}
           \FOR{$l=1,2,\dots,...$}
        \STATE Update $\bs_i$ by \eqref{eq:update_s}, $\bq_i$ by minimizing \eqref{eq:update_q}, $\bdelta_i$ by \eqref{eq:update_delta}, $\bv_1$ by \eqref{eq:updatev1} and $\bv_2$ by \eqref{eq:updatev2}.
        \STATE Stop and get $\tilde{\bs}_i$ and $\hat{\bq}_i$ if convergence criterion is met.
          \ENDFOR
		\ENDFOR
  \STATE Obtain $\hat{\bQ} = (\hat{\bq}_1,\dots, \hat{\bq}_r)$.
	\end{algorithmic}
\end{algorithm}


\begin{remark}
  The stopping criterion is $\Vert \bs - \bB \bq\Vert_2\leq \delta_e$ as in the literature \citep{ma2017concave,wang2023spatial, liu::wang2025}, where $\delta$ is a small positive value. Here we use $\delta_e =  10^{-5}$.
\end{remark}

\begin{remark}
  We use the Bayesian Information Criterion (BIC) to select the tuning parameters. The BIC is defined as 
  \begin{equation}
      BIC(\lambda_1, \lambda_2) = \log(\frac{1}{np} \sum_{t=1}^n\Vert \bx_t - \hat{\bx}_t  \Vert^2 ) + \frac{\log(np)}{np}
    \vert \hat{\bQ} (\lambda_1, \lambda_2) \vert,
  \end{equation}
  where $\hat{\bx}_t = \hat{\bQ} (\hat{\bQ}^\top \hat{\bQ})^{-1} \hat{\bQ}^\top\bx_t$, and $\vert \hat{\bQ} (\lambda_1, \lambda_2) \vert$ is the number of nonzero elements in $\hat{\bQ} (\lambda_1, \lambda_2)$. We adopt a two-step procedure to select the tuning parameters as in the literature \citep{tang2023multivariate, zhang2025simultaneously}. First, we set $\lambda_2 =0$ and select $\lambda_1$ from a sequence of candidate values that minimizes the BIC. Then, holding $\lambda_1$ fixed at the selected value from step 1, we choose $\lambda_2$ from a sequence of candidate values with the smallest BIC to obtain the final estimate. 
\end{remark}

\section{Theoretical properties}
\label{sec:them}

In this section, we will study the asymptotic properties of our proposed estimator.

We use the same regularity conditions as those (C\ref{cond_alphamix})-(C\ref{cond_coherence}) in \cite{liu::wang2025}, which are listed below.

\begin{enumerate}
\renewcommand{\labelenumi}{\textbf{(C\arabic{enumi})}}

\item  Let $\calF_i^j$ be the $\sigma$-field generated by $\{\bff_t^s: i \leq t \leq j\}$.
The joint process $\{\bff_t^s\}$ is
$\alpha$-mixing with mixing coefficients satisfying
$
\sum_{t=1}^{\infty} \alpha(t)^{1-2/\gamma}<\infty,
$
for some $\gamma>2$, where
$\alpha(t)=\sup_{i} \sup_{A \in \calF_{-\infty}^i, B \in \calF_{i+t}^{\infty}} |P(A \cap B) -P(A)P(B)|$.  \label{cond_alphamix}

\item For any $i=1,\ldots, r$, $t=1, \ldots, n$, $E(|f_{t,i}^s|^{2\gamma})< \sigma_f^{2\gamma}$, where $f_{t,i}^s$ is the $i$-th element of $\bff_t^s$,  $\sigma_f>0$ is a constant, and $\gamma$ is given in Condition (C\ref{cond_alphamix}).  \label{cond_fbound}

\item  $\bve_{t}$ and $\bff_t^s$ are uncorrelated given ${\cal F}_{-\infty}^{t-1}$. Let $\bSigma_{e,t}$ be the covariance of $\bve_t$. $|\sigma_{e,t,ij}|<\bSigma_{\epsilon}^2 <\infty$ for $i,j=1,\ldots,p$, and $t=1,\ldots,n$. In other words, the absolute value of each element of $\bSigma_{e,t}$ remains bounded by a constant $\sigma_{\epsilon}^2$ as $p$ increases to infinity, for $t=1,\ldots, n$. \label{cond_cov}


\item  There exists a constant $\delta \in [0,1]$ such that $\|\bA^s\|_2^2  \asymp \|\bA^s\|^2_{\min}\asymp m^{1-\delta}$, as $p$ goes to infinity, where $m=\sum_{i=1}^r m_i$ is the number of nonzero elements in $\bA^s$. Furthermore, $\Vert \bA^s\Vert_{\max} \leq C_1$, where $C_1$ is a positive constant. In addition, $m_1 \asymp m_2 \asymp \ldots \asymp m_r \asymp m$. \label{cond_strength}

\item  $\bM$ has $r$ distinct nonzero eigenvalues. \label{cond_eigenM}

\item $\bve_t$'s are independent sub-Gaussian random vectors. Each random vector in the sequences $\bff_t^s$ follows a sub-Gaussian distribution. \label{cond_subgaussian}
\end{enumerate}

Two primary strategies exist for distinguishing between the noise component and the latent factors. One approach assumes that the idiosyncratic errors exhibit both weak temporal and weak cross-sectional dependence with $\sum_{i=1}^p\sum_{j=1}^p |\sigma_{e,t,ij}|\leq Cp$ for any $t=1,\ldots,n$, where $C$ is a positive constant; see \cite{bai2002}, \cite{bai2003}, \cite{bai2006evaluating}, \cite{bai2008forecasting}, \cite{uematsu2022estimation}, \cite{uematsu2022inference} and among others. The alternative assumes that the noise process is serially uncorrelated, but allows for strong cross-sectional dependence with $|\sigma_{e,t,ij}|<C$ for any $i,j=1,\ldots,p$ and $t=1,\ldots,n$ \citep{lam2011,lam2012,chang2015,wang2019,chen2022}. In this paper, we adopt the latter assumption. Nonetheless, we believe that our framework can be extended to accommodate the former, which we leave as a direction for future research.

Conditions (C\ref{cond_alphamix})--(C\ref{cond_cov}) and Condition (C\ref{cond_eigenM}) are standard assumptions in the literature in factor models \citep{lam2011, lam2012, chang2015, liu2016,wang2019,liu2022} and used to ensure that the estimated autocovariance matrices converge. (C\ref{cond_strength}) gives the assumption of the strength of factors, similar to that in \cite{chang2015}. \cite{liu::wang2025} has a detailed discussion about the role of $m$ and $p$. Condition (C\ref{cond_subgaussian}) is a commonly used assumption in models for high-dimensional data analysis, such as regression models in \cite{ma2017concave} and \cite{wang2023spatial}, and factor models for functional time series, as in \cite{guo2021factor} and \cite{fang2022finite}.

Recall that $\bQ$ is not necessarily an orthogonal matrix. Hence, we impose an assumption to ensure that column vectors in $\bQ$ are well separated as the dimension grows. To achieve this, we first obtain the orthogonal basis of $\cM(\bQ)$ using Gram-Schmidt orthonormalization. Specifically, let $\bS=(\bs_1, \bs_2,\ldots, \bs_r)$, where $\bs_1=\bq_1$, and
$
\bs_i=(\bI-\bS_i\bS_i^\top) \bq_i, 
$
where $\bS_i=(\bs_1,\ldots, \bs_{i-1})$ for $i=2,\ldots, r$. Let $\mathcal{V}_{{i}}$
denote the nonzero indices of ${\bf q}_{i}$ and $\mathcal{V}_{s_{i}}$
denote the nonzero indices of ${\bf s}_{i}$.
We define $\mathcal{V}_{i}^*=\mathcal{V}_{s_{1}}\cup\mathcal{V}_{s_{2}}\dots\cup\mathcal{V}_{s_{i-1}}\cup\mathcal{V}_{i}$, and $\mathcal{N}_{i}^*=\mathcal{V}_{i}^*\backslash\mathcal{V}_{i}$. 
$\mathcal{N}_i^*$ contains indices where the corresponding elements in $\bq_i$ are zero while the corresponding elements in at least one of $\{\bs_i\mid i=1,\ldots, i-1\}$ are nonzero. Note that $\mathcal{N}_i^*$ cannot be an empty set. Otherwise, $(\mathcal{V}_{s_{1}}\cup\mathcal{V}_{s_{2}}\dots\cup\mathcal{V}_{s_{i-1}}) \subset \mathcal{V}_i$, which means that there exists a vector $\bv \in \mathbb{R}^{(i-1)}$ such that $(\bq_i-\bS_i\bv)$ is more sparse than $\bq_i$ and thus $(\bq_1, \ldots, \bq_i-\bS_i\bv)$ is more sparse than $(\bq_1,\ldots,\bq_i)$. If that is true, $\bQ$ would not be one of the loading matrices with most zero elements in $\cM(\bA)$.

Let ${\bf S}_{i,1}={\bf S}_{i[\mathcal{N}_{i}^{*}]}$, we also have the following two assumptions about $\bS$ and one assumption about the group sparsity.  

\begin{enumerate}
\renewcommand{\labelenumi}{\textbf{(C\arabic{enumi})}}
\setcounter{enumi}{6}
\item $\Vert\mathbf{S}_{i,1}\Vert_{\min}
\asymp 1$. \label{cond_si1_eigen} 

\item  There exists a positive constant $C_{\mu} > 1$ such that $\Vert{\bf S}\Vert_{2\rightarrow\infty}\leq C_{\mu}\sqrt{\frac{r}{m}}$. \label{cond_coherence} 


\end{enumerate}

Condition (C\ref{cond_si1_eigen}) indicates that the column vectors in $\bQ$ are far apart and each column vector provides enough information about zero elements as the dimension grows. The bounded coherence assumption in Condition (C\ref{cond_coherence}) is widely used in matrix theory; see examples in \cite{fan2018pertur} and \cite{cape2019two}. \cite{cape2019two} assumes $\Vert\bS\Vert_{2\rightarrow} \leq C_{\mu}\sqrt{\frac{r}{p}}$ for a $p \times r$ orthonormal matrix $\bS$. Since $\bQ$ in our setting is sparse with $m$ nonzero elements, we replace $p$ with $m$ and assume that the sparsity level of $\bS$ is $O(m)$. 

Let $\mathcal{V}_i$ be the index of nonzero elements of ${\bf q}_i$ and $\mathcal{V}^{g}_{i}$ be the index of nonzero groups of ${\bf q}_i$, which is a subset of $\{1,2,\dots, J_i\}$. 
Let $b_1 = \min_i\min_{j\in \mathcal{V}_i} \vert q_{ij}\vert$ and $b_2 = \min_i \min_{j^\prime \in \mathcal{V}^{g}_{i}} \frac{1}{\sqrt{d_{ij^\prime}}} \Vert {\bf q}_{i(j^\prime)}\Vert$, which define minimal signals.  Denote $\phi_{n,p,m} = \max\left(m^{2\delta-2}p^{2}n^{-1/2},m^{\delta}\right)$ and define $\tau_{n,p,m}=\phi_{n,p,m}\sqrt{\frac{\log p}{n}}$ if $m=o(p)$, and $\tau_{n,p,m}=p^{\delta}n^{-1/2}$ if $m=O(p)$. We have the following result for the proposed estimator of the loading matrix. The proof is provided in Appendix B.

\begin{theorem}
\label{thm_est}
Assume that $m_1 < m_2 < \ldots < m_r$, $b_1 > \gamma\lambda_1$ and $b_2 > \gamma\lambda_2$. If $\lambda_1 \gtrsim \tau_{n,p,m}$, $\lambda_2 \gtrsim \tau_{n,p,m}$ and $\tau_{n,p,m}= o(1)$ as $n\rightarrow \infty$ and $p\rightarrow\infty$. Under Conditions (C\ref{cond_alphamix})-(C\ref{cond_si1_eigen}), then 
\begin{align*}
    &\Vert\hat{{\bf Q}}-{\bf Q}\Vert_2= O(\tau_{n,p,m})= \begin{cases}
      O_{p}\left(\phi_{n,p,m}\sqrt{\frac{\log p}{n}}\right) & \text{if } m=o(p), \\  
      O_{p}\left(m^{\delta - 1} p n^{-1/2}\right) = O_{p}\left(p^\delta n^{-1/2}\right) & \text{if } m = O(p),\\
    \end{cases} \\
  & P (\hat{\mathcal{V}}_i = \mathcal{V}_i) = 1, \mbox{ for } i=1,\ldots, r,
\end{align*}
as $n$ and $p$ go to infinity, where $\hat{\mathcal{V}}_i$ contains the indices of nonzero elements in $\hat{\bq}_i$.
\end{theorem}

Theorem \ref{thm_est} shows that the proposed estimator is consistent under some regularity conditions. It also reveals that the estimation error depends on the sparsity level of the loading matrix. It converges to zero as fast as the estimator proposed in \cite{liu::wang2025}. If $m$ has the same order as $p$, $\hat{\bQ}$ converges at the same rate as the estimator proposed in \cite{lam2012}. If $m =o(p)$, the convergence rate of $\bQ$ is determined by two terms; the first one is the squared bias and the second one is variance. When $\delta <1$ and the loading matrix is quite sparse, the estimation error is dominated by the first term, $O(m^{2\delta-2}p^2n^{-1/2}\sqrt{\frac{\log p}{n}})$; the more sparse the loading matrix is, the larger the bias is. When $\delta <1$ and the loading matrix is not very sparse, the error is dominated by $O(m^\delta \sqrt{\frac{\log p}{n}})$. The more sparse the loading matrix, the smaller the variance is. Moreover, the results also indicate that the $\mathcal{V}_i^g$ can be recovered with probability approaching 1. 

\begin{remark}
We do not impose specific assumptions on $J_i$, ${\cal G}_i$, $\{d_{ij}\}$, or the group sparsity structure. The number of groups and sizes of groups can be fixed or can grow to infinity as $p$ grows. These conditions do not influence the convergence rate of the proposed estimator. This is because the tuning parameter $\lambda_2$ is multiplied by $\sqrt{d_{ij}}$ in the objective functions \eqref{eq:step1} and \eqref{eq:step2} to eliminate the unbalanced impacts of different group sizes have on the penalty term. 
\end{remark}

The result in Theorem \ref{thm_est} implies the following Corollary. 
\begin{corollary}
    \label{thm_estimate}
    If all eigenvalues of ${\bf \Sigma}_{e,t}$ are uniformly bounded from infinity as $p\rightarrow \infty$, it holds that 
    \begin{equation}
        p^{-1/2}\Vert\hat{{\bf Q}}\hat{{\bf z}}_{t}-{\bf Q}{\bf z}_{t}\Vert_{2}=O_{p}\left(p^{-1/2}m^{1/2-\delta/2}\Vert\hat{{\bf Q}}-{\bf Q}\Vert_{2}+p^{-1/2}\right),
    \end{equation}
    as $n$ and $p$ go to infinity.
\end{corollary}
 Corollary \ref{thm_estimate} indicates that the estimated common component is also consistent.

\section{Simulation Study}
\label{sec:sim}

In this section, we present several simulated examples to evaluate the performance of the proposed approach and compare it with other existing approaches. We consider a scenario where all loadings have the same group structure in Section \ref{subsec_example1}, and consider a scenario with two distinct group structures in Section \ref{subsec_example2}.


To evaluate estimation accuracy, we report the estimation error of the loading space for a fair comparison, $\cD({\cM}(\bA^s), {\cM}(\hat{\bQ}))$, as defined in (\ref{eq:distance}). To assess the performance of identifying sparsity, we report the false negative value (``FN", the number of elements falsely identified as zero), false positive (``FP", the number of elements falsely identify as nonzero), and F1 score, which is a number between 0 and 1 and measures the classification accuracy (nonzero or zero). A higher F1 score indicates better identification performance. 

We consider three methods in our comparison. ``eigen'' refers to the standard method by \cite{lam2011}, presented in Section \ref{subsec:reviewlam}, ``sparse'' refers to the method in \cite{liu::wang2025}, which consider individual sparsity in the loading matrix but not group sparsity; and ``sparsegroup" refers to our proposed approach, which accommodates both individual and group sparsity. We compare these approaches for $p = 60, 120, 200$ and $n = 50, 100, 200, 500$. Note that the sparsity structures of $\bA^s$ and $\bQ$ are the same.

Datasets are simulated from the model \eqref{eq::model112}. The nonzero elements of $\bA^s$ are drawn from a truncated standard normal distribution, with absolute values bounded above by 0.1. We set $r=3$, and the factor process $\mathbf{f}_t^s$ is generated from three independent AR(1) processes, each with an AR coefficient of $0.9$ and an innovation variance of 1. The diagonal elements of $\bSigma_{e,t}$ are all set to 1, and the off-diagonal elements are set to 0.5. The number of factors is assumed to be known. For each setting, we generate 500 samples and compare the estimation results.

\subsection{Example 1}
\label{subsec_example1}

In $\bA^s$, each column consists of five groups and shares the same group structure, i.e., ${\cal G}_1={\cal G}_2={\cal G}_3$. The group sizes $d_{ij^\prime}$ are $p/6$, $p/6$, $p/6$, $p/4$ and $p/4$, respectively. In $\ba_1^s$, the first three groups are nonzero, with the last $1/3$ of the elements in each group set to zero. In $\ba_2^s$, the first four groups are nonzero, with the first $1/3$ of the elements in each group set to zero. In $\ba_3^s$, the last four groups are nonzero. Figure \ref{fig:sparse_structure} illustrates the transpose of the loading matrix when $p=60$. In the figure, the first, second, and third rows represent the sparsity structures of $\ba_1^s$, $\ba_2^s$, and $\ba_3^s$, respectively.  

\begin{figure}[H]
    \centering
    \includegraphics[scale=1]{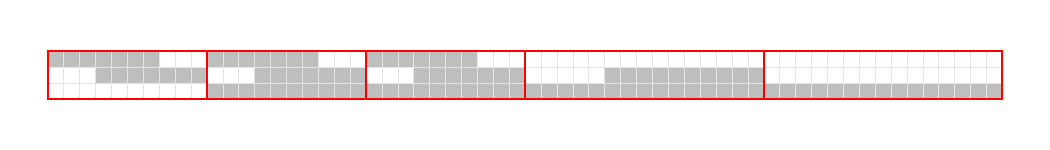}
    \caption{Sparse structure of the transpose of $\bA^s$ in Example 1. Grey cells represent the elements that are nonzero, while white cells represent the elements that are zero. }
    \label{fig:sparse_structure}
\end{figure}



Table \ref{tab:dist} reports the distance between the estimated loading space and the true loading space for $p = 60, 120, 200$ and $n = 50, 100, 200, 500$. It can be observed that both the sparse factor model and the sparse group factor model outperform the standard approach. The proposed method estimates the loading matrix more accurately than the method of \cite{liu::wang2025}, particularly when the sample size is small. As the sample size increases, the advantage of the proposed method diminishes, which confirms that the proposed estimator converges as fast as the one proposed in \cite{liu::wang2025}.

\begin{table}[H]
\centering
\caption{Distance between the estimated loading space and the true loading space under Example 1}
\label{tab:dist}
\begin{tabular}{lllll}
  \hline
$p$ & $n$ & eigen & sparse & sparsegroup \\ 
  \hline
60 & 50 & 0.179(0.078) & 0.168(0.081) & 0.155(0.080) \\ 
  60 & 100 & 0.083(0.026) & 0.074(0.026) & 0.068(0.023) \\ 
  60 & 200 & 0.045(0.013) & 0.038(0.012) & 0.036(0.011) \\ 
  60 & 500 & 0.024(0.006) & 0.020(0.006) & 0.019(0.005) \\ 
  \hline
  120 & 50 & 0.175(0.065) & 0.162(0.067) & 0.147(0.064) \\ 
  120 & 100 & 0.083(0.025) & 0.072(0.024) & 0.066(0.022) \\ 
  120 & 200 & 0.045(0.013) & 0.037(0.011) & 0.035(0.010) \\ 
  120 & 500 & 0.024(0.006) & 0.019(0.004) & 0.018(0.004) \\ 
  \hline
  200 & 50 & 0.178(0.073) & 0.164(0.076) & 0.150(0.073) \\ 
  200 & 100 & 0.081(0.025) & 0.070(0.024) & 0.065(0.021) \\ 
  200 & 200 & 0.044(0.012) & 0.036(0.010) & 0.034(0.009) \\ 
  200 & 500 & 0.024(0.005) & 0.018(0.005) & 0.018(0.004) \\ 
   \hline
\end{tabular}
\end{table}

Tables \ref{tab:ex1_60}, \ref{tab:ex1_120}, and \ref{tab:ex1_200} report the false negatives (``FN", the number of elements falsely identified as zero), false positives (``FP", the number of elements falsely identified as nonzero), and F1 score for various setups under the sparse structure shown in Figure \ref{fig:example}. We observe that the sparse group factor model outperforms the sparse factor model in identifying nonzero elements, particularly when the sample sizes are not large. 

\begin{table}[H]
\centering
\caption{Mean and standard deviation (in parentheses) for sparsity identification when $p= 60$}
\label{tab:ex1_60}
{\scriptsize
\begin{tabular}{ll|ll|ll|ll|ll}
  \hline
& & \multicolumn{2}{c|}{$n=50$} & \multicolumn{2}{c|}{$n=100$} & \multicolumn{2}{c|}{$n=200$} & \multicolumn{2}{c}{$n=500$} \\ 
  \hline
& & sparse & sparsegroup & sparse & sparsegroup & sparse & sparsegroup & sparse & sparsegroup \\ 
\hline
\multirow{3}{*}{loadings 1} & FN & 2.49 (2.81) & 2.04 (2.59) & 1.01 (1.70) & 0.86 (1.62) & 0.35 (0.88) & 0.32 (0.85) & 0.13 (0.90) & 0.13 (0.90) \\ 
 & FP & 10.9 (7.73) & 6.95 (7.77) & 8.09 (7.46) & 3.84 (5.52) & 3.44 (4.91) & 1.66 (2.90) & 2.30 (4.70) & 1.29 (3.07) \\ 
 & F1 & 0.75 (0.15) & 0.82 (0.16) & 0.83 (0.13) & 0.90 (0.11) & 0.92 (0.09) & 0.96 (0.06) & 0.95 (0.09) & 0.97 (0.07) \\ 
  \hline
 \multirow{3}{*}{loadings 2} & FN & 4.08 (3.37) & 3.71 (3.66) & 1.86 (2.20) & 1.72 (2.48) & 0.88 (1.33) & 0.80 (1.44) & 0.51 (1.88) & 0.56 (1.91) \\ 
  &FP & 8.68 (7.26) & 7.10 (7.12) & 4.37 (5.53) & 3.16 (4.30) & 2.28 (3.64) & 1.76 (2.68) & 1.32 (2.84) & 1.86 (3.28) \\ 
  &F1 & 0.82 (0.13) & 0.84 (0.13) & 0.91 (0.09) & 0.93 (0.08) & 0.95 (0.06) & 0.96 (0.05) & 0.97 (0.07) & 0.96 (0.07) \\ 
  \hline
  \multirow{3}{*}{loadings 3} & FN & 7.17 (3.32) & 5.51 (3.36) & 4.02 (2.33) & 2.93 (2.04) & 2.39 (1.77) & 1.62 (1.45) & 1.30 (1.35) & 0.66 (0.89) \\ 
 & FP & 2.41 (2.48) & 0.96 (2.43) & 1.45 (2.10) & 0.50 (1.34) & 1.30 (2.33) & 0.41 (0.92) & 1.91 (3.01) & 0.76 (1.48) \\ 
  &F1 & 0.90 (0.05) & 0.93 (0.05) & 0.94 (0.04) & 0.96 (0.03) & 0.96 (0.03) & 0.98 (0.02) & 0.97 (0.04) & 0.99 (0.02) \\ 
   \hline
\end{tabular}
}
\end{table}

\begin{table}[H]
\centering
\caption{Mean and standard deviation (in parentheses) for sparsity identification when $p= 120$}
\label{tab:ex1_120}
{\scriptsize
\begin{tabular}{ll|ll|ll|ll|ll}
  \hline
& & \multicolumn{2}{c|}{$n=50$} & \multicolumn{2}{c|}{$n=100$} & \multicolumn{2}{c|}{$n=200$} & \multicolumn{2}{c}{$n=500$} \\ 
  \hline
& & sparse & sparsegroup & sparse & sparsegroup & sparse & sparsegroup & sparse & sparsegroup \\ 
\hline
\multirow{3}{*}{loadings 1} & FN & 4.59 (5.43) & 3.85 (5.34) & 1.78 (3.21) & 1.48 (3.08) & 0.50 (0.77) & 0.41 (0.68) & 0.09 (0.32) & 0.08 (0.30) \\ 
 & FP & 21.7 (15.7) & 13.4 (15.1) & 14.7 (14.4) & 6.93 (9.72) & 7.79 (10.6) & 3.56 (4.56) & 2.85 (4.64) & 2.16 (2.67) \\ 
&  F1 & 0.74 (0.16) & 0.82 (0.17) & 0.84 (0.14) & 0.91 (0.11) & 0.91 (0.09) & 0.95 (0.05) & 0.97 (0.05) & 0.97 (0.03) \\
  \hline
  \multirow{3}{*}{loadings 2} & FN & 7.36 (5.79) & 6.60 (7.00) & 3.17 (3.53) & 2.78 (4.13) & 1.16 (1.15) & 0.95 (1.04) & 0.24 (0.54) & 0.23 (0.51) \\ 
 & FP & 16.7 (15.0) & 13.6 (13.5) & 7.71 (10.2) & 5.59 (7.54) & 2.98 (5.64) & 2.21 (3.40) & 0.81 (2.03) & 0.97 (2.55) \\ 
 & F1 & 0.82 (0.13) & 0.84 (0.13) & 0.92 (0.09) & 0.93 (0.08) & 0.97 (0.04) & 0.97 (0.03) & 0.99 (0.02) & 0.99 (0.02) \\ 
  \hline
  \multirow{3}{*}{loadings 3} & FN & 13.9 (6.00) & 10.7 (6.13) & 6.99 (3.32) & 5.52 (2.93) & 3.81 (2.39) & 2.96 (2.10) & 1.39 (1.46) & 0.91 (1.18) \\ 
 & FP & 3.32 (4.15) & 1.02 (3.48) & 1.40 (2.45) & 0.39 (1.36) & 0.70 (1.99) & 0.43 (1.04) & 0.90 (3.00) & 0.70 (1.30) \\ 
&  F1 & 0.91 (0.04) & 0.94 (0.05) & 0.96 (0.02) & 0.97 (0.02) & 0.98 (0.02) & 0.98 (0.01) & 0.99 (0.02) & 0.99 (0.01) \\ 
   \hline
\end{tabular}
}
\end{table}

\begin{table}[H]
\centering
\caption{Mean and standard deviation (in parentheses) for sparsity identification when $p= 200$}
\label{tab:ex1_200}
{\scriptsize
\begin{tabular}{ll|ll|ll|ll|ll}
  \hline
& & \multicolumn{2}{c|}{$n=50$} & \multicolumn{2}{c|}{$n=100$} & \multicolumn{2}{c|}{$n=200$} & \multicolumn{2}{c}{$n=500$} \\ 
  \hline
& & sparse & sparsegroup & sparse & sparsegroup & sparse & sparsegroup & sparse & sparsegroup \\ 
\hline
\multirow{3}{*}{loadings 1} & FN & 8.15 (9.80) & 6.82 (9.32) & 2.61 (4.36) & 2.21 (4.22) & 0.85 (0.99) & 0.65 (0.89) & 0.13 (0.43) & 0.10 (0.39) \\ 
  &FP & 35.1 (25.1) & 20.5 (23.7) & 21.6 (21.1) & 9.40 (12.2) & 12.2 (16.1) & 5.56 (5.90) & 5.14 (11.3) & 2.85 (4.00) \\ 
  & F1 & 0.74 (0.17) & 0.83 (0.17) & 0.85 (0.12) & 0.92 (0.09) & 0.92 (0.09) & 0.96 (0.04) & 0.97 (0.06) & 0.98 (0.03) \\ 
  \hline
 \multirow{3}{*}{loadings 2} &FN & 12.8 (10.0) & 11.4 (12.2) & 4.83 (3.78) & 4.05 (4.71) & 1.91 (1.59) & 1.47 (1.38) & 0.38 (0.72) & 0.34 (0.69) \\ 
  & FP & 28.3 (24.8) & 22.0 (22.3) & 12.01 (16.3) & 9.09 (11.9) & 4.04 (7.54) & 3.36 (5.07) & 1.41 (2.73) & 1.38 (2.95) \\ 
  & F1 & 0.82 (0.13) & 0.84 (0.13) & 0.92 (0.07) & 0.94 (0.06) & 0.97 (0.03) & 0.98 (0.03) & 0.99 (0.01) & 0.99 (0.01) \\ 
  \hline
  \multirow{3}{*}{loadings 3} &FN & 22.1 (8.87) & 17.3 (10.2) & 10.9 (4.92) & 8.71 (4.04) & 5.63 (3.07) & 4.32 (2.71) & 1.62 (1.61) & 1.35 (1.45) \\ 
  & FP & 5.34 (6.67) & 1.71 (6.17) & 1.63 (3.06) & 0.42 (1.23) & 0.62 (1.86) & 0.63 (1.32) & 0.19 (1.26) & 1.09 (1.83) \\ 
  & F1 & 0.91 (0.04) & 0.94 (0.05) & 0.96 (0.02) & 0.97 (0.01) & 0.98 (0.01) & 0.98 (0.01) & 0.99 (0.01) & 0.99 (0.01) \\ 
   \hline
\end{tabular}
}
\end{table}

\subsection{Example 2}
\label{subsec_example2}

In this example, we have ${\cal G}_1={\cal G}_2 \neq {\cal G}_3$, with $J_1 = J_2 = 5$ and $J_3 = 6$. $\ba_1^s$ and $\ba_2^s$ have the same group structure, while the group structure of $\ba_3^s$ differs from the other two. For $\ba_1^s$ and $\ba_2^s$, the group sizes $d_{ij^\prime}$ are $p/6$, $p/6$, $p/6$, $p/4$ and $p/4$. For $\ba_3^s$, each group has a size of $p/6$. Figure \ref{fig:example2} illustrates the sparsity structure of the transpose of the loading matrix when $p=60$. In the first and second rows, different groups are separated by red lines, while in the third row, different groups are separated by blue lines.

\begin{figure}[H]
    \centering
    \includegraphics[scale =1]{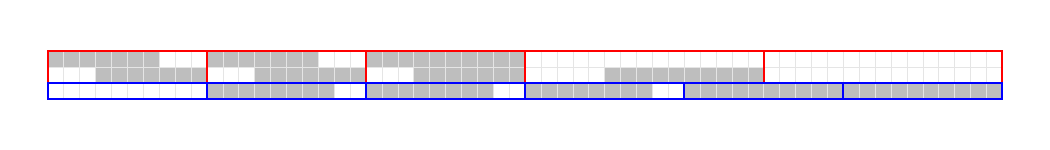}
    \caption{Sparse structure of the transpose of $\bA^s$ in Example 1. Grey cells represent the elements that are nonzero, while white cells represent the elements that are zero.}
    \label{fig:example2}
\end{figure}


Table \ref{tab:dist2} presents the distance between the estimated loading space and the true loading space for $p = 60, 120, 200$ and $n = 50, 100, 200, 500$ under the sparse structure shown in Figure \ref{fig:example2}. Tables \ref{tab:ex2_60}, \ref{tab:ex2_120}, \ref{tab:ex2_200} report ``FN", ``FP" and F1 score for comparing the identification of nonzero elements. We observe similar patterns to those in Section \ref{subsec_example2}. In terms of estimating loading space and identifying nonzero elements, the proposed sparse group factor model outperforms both the standard method and the method of \cite{liu::wang2025}, which does not use the group information.

\begin{table}[ht]
\centering
\caption{Distance between the estimated loading space and the true loading space under Example 2}
\label{tab:dist2}
\begin{tabular}{lllll}
  \hline
$p$ & $n$ & eigen & sparse & sparsegroup \\ 
  \hline
60 & 50 & 0.179(0.078) & 0.167(0.081) & 0.155(0.080) \\ 
  60 & 100 & 0.083(0.026) & 0.073(0.026) & 0.068(0.023) \\ 
  60 & 200 & 0.043(0.012) & 0.037(0.011) & 0.035(0.010) \\ 
  60 & 500 & 0.024(0.006) & 0.019(0.005) & 0.019(0.005) \\ 
  \hline
  120 & 50 & 0.176(0.071) & 0.161(0.073) & 0.149(0.070) \\ 
  120 & 100 & 0.081(0.023) & 0.069(0.021) & 0.064(0.019) \\ 
  120 & 200 & 0.044(0.011) & 0.036(0.010) & 0.034(0.009) \\ 
  120 & 500 & 0.023(0.005) & 0.018(0.004) & 0.018(0.004) \\ 
  \hline
  200 & 50 & 0.176(0.074) & 0.162(0.077) & 0.149(0.075) \\ 
  200 & 100 & 0.082(0.026) & 0.071(0.024) & 0.066(0.021) \\ 
  200 & 200 & 0.043(0.012) & 0.036(0.010) & 0.034(0.009) \\ 
  200 & 500 & 0.023(0.005) & 0.018(0.004) & 0.018(0.004) \\ 
   \hline
\end{tabular}
\end{table}

\begin{table}[H]
\centering
\caption{Mean and standard deviation (in parentheses) for sparsity identification when $p= 60$}
\label{tab:ex2_60}
{\scriptsize
\begin{tabular}{ll|ll|ll|ll|ll}
  \hline
& & \multicolumn{2}{c|}{$n=50$} & \multicolumn{2}{c|}{$n=100$} & \multicolumn{2}{c|}{$n=200$} & \multicolumn{2}{c}{$n=500$} \\ 
  \hline
&  & sparse & sparsegroup & sparse & sparsegroup & sparse & sparsegroup & sparse & sparsegroup \\ 
\hline
\multirow{3}{*}{loadings 1} & FN & 3.60 (3.45) & 2.97 (3.21) & 2.00 (2.96) & 1.72 (2.79) & 1.04 (2.41) & 0.98 (2.36) & 0.58 (2.04) & 0.56 (1.97) \\ 
 & FP & 9.44 (7.10) & 6.13 (7.16) & 6.15 (6.51) & 3.32 (5.53) & 3.37 (5.15) & 2.01 (4.32) & 1.88 (4.34) & 1.30 (3.65) \\ 
 & F1 & 0.77 (0.16) & 0.83 (0.16) & 0.85 (0.15) & 0.91 (0.14) & 0.92 (0.12) & 0.94 (0.12) & 0.96 (0.11) & 0.97 (0.10) \\ 
 \hline
 \multirow{3}{*}{loadings 2} & FN & 5.20 (4.45) & 4.78 (4.79) & 2.88 (3.84) & 2.73 (4.06) & 1.77 (3.77) & 1.71 (3.86) & 1.18 (3.55) & 1.15 (3.46) \\ 
  & FP & 9.70 (7.66) & 8.29 (7.51) & 5.08 (5.95) & 3.84 (4.82) & 2.57 (4.29) & 2.20 (3.46) & 1.89 (3.67) & 2.02 (3.53) \\ 
  & F1 & 0.78 (0.15) & 0.81 (0.15) & 0.88 (0.13) & 0.90 (0.12) & 0.93 (0.12) & 0.94 (0.11) & 0.95 (0.11) & 0.95 (0.11) \\ 
  \hline
 \multirow{3}{*}{loadings 3} & FN & 5.92 (3.83) & 4.53 (3.50) & 3.08 (2.71) & 2.44 (2.48) & 1.54 (1.77) & 1.15 (1.77) & 0.64 (1.50) & 0.42 (1.31) \\ 
  &FP & 3.71 (4.23) & 2.97 (4.16) & 1.49 (2.50) & 1.03 (1.85) & 0.98 (2.46) & 0.59 (1.34) & 1.13 (3.10) & 0.54 (1.25) \\ 
 & F1 & 0.89 (0.08) & 0.91 (0.08) & 0.95 (0.05) & 0.96 (0.05) & 0.97 (0.04) & 0.98 (0.03) & 0.98 (0.05) & 0.99 (0.03) \\ 
   \hline
\end{tabular}
}
\end{table}

\begin{table}[H]
\centering
\caption{Mean and standard deviation (in parentheses) for sparsity identification when $p= 120$}
\label{tab:ex2_120}
{\scriptsize
\begin{tabular}{ll|ll|ll|ll|ll}
  \hline
& & \multicolumn{2}{c|}{$n=50$} & \multicolumn{2}{c|}{$n=100$} & \multicolumn{2}{c|}{$n=200$} & \multicolumn{2}{c}{$n=500$} \\ 
  \hline
& & sparse & sparsegroup & sparse & sparsegroup & sparse & sparsegroup & sparse & sparsegroup \\ 
\hline 
\multirow{3}{*}{loadings 1} & FN & 7.27 (7.26) & 6.08 (6.97) & 2.85 (4.51) & 2.48 (4.38) & 1.16 (3.26) & 1.02 (3.16) & 0.35 (2.16) & 0.34 (2.12) \\ 
 & FP & 16.4 (14.0) & 10.7 (13.8) & 10.6 (11.5) & 4.44 (7.48) & 4.87 (7.66) & 2.26 (5.22) & 1.74 (4.13) & 1.24 (3.67) \\ 
 & F1 & 0.78 (0.17) & 0.84 (0.18) & 0.87 (0.12) & 0.93 (0.11) & 0.94 (0.09) & 0.97 (0.08) & 0.98 (0.06) & 0.98 (0.05) \\ 
 \hline
 \multirow{3}{*}{loadings 2} &  FN & 10.2 (9.24) & 9.86 (10.3) & 3.97 (4.86) & 3.64 (5.52) & 1.78 (4.24) & 1.62 (4.45) & 0.62 (3.27) & 0.62 (3.28) \\ 
 & FP & 18.3 (15.6) & 15.0 (14.0) & 7.33 (10.4) & 5.30 (6.85) & 3.05 (5.87) & 2.47 (4.38) & 1.04 (3.42) & 1.33 (3.63) \\ 
 & F1 & 0.78 (0.17) & 0.80 (0.17) & 0.91 (0.10) & 0.93 (0.09) & 0.96 (0.08) & 0.97 (0.07) & 0.99 (0.05) & 0.98 (0.06) \\ 
 \hline
 \multirow{3}{*}{loadings 1} & FN & 11.5 (6.36) & 9.26 (7.45) & 5.37 (2.79) & 4.29 (2.56) & 2.63 (2.01) & 2.07 (1.73) & 0.72 (1.02) & 0.62 (0.94) \\ 
 & FP & 6.87 (8.58) & 5.96 (8.25) & 1.86 (3.43) & 1.42 (2.45) & 0.59 (1.74) & 0.71 (1.41) & 0.40 (2.01) & 0.70 (1.23) \\ 
 & F1 & 0.89 (0.07) & 0.91 (0.09) & 0.96 (0.03) & 0.97 (0.02) & 0.98 (0.02) & 0.98 (0.01) & 0.99 (0.01) & 0.99 (0.01) \\ 
   \hline
\end{tabular}
}
\end{table}

\begin{table}[H]
\centering
\caption{Mean and standard deviation (in parentheses) for sparsity identification when $p= 200$}
\label{tab:ex2_200}
{\scriptsize
\begin{tabular}{ll|ll|ll|ll|ll}
  \hline
& & \multicolumn{2}{c|}{$n=50$} & \multicolumn{2}{c|}{$n=100$} & \multicolumn{2}{c|}{$n=200$} & \multicolumn{2}{c}{$n=500$} \\ 
  \hline
& & sparse & sparsegroup & sparse & sparsegroup & sparse & sparsegroup & sparse & sparsegroup \\ 
\hline  
\multirow{3}{*}{loadings 1} & FN & 11.9 (11.5) & 9.85 (11.0) & 4.16 (6.00) & 3.51 (5.79) & 1.64 (3.68) & 1.42 (3.61) & 0.34 (2.21) & 0.31 (2.13) \\ 
  &FP & 27.6 (22.0) & 17.0 (21.5) & 17.0 (19.1) & 6.85 (12.7) & 7.81 (13.0) & 3.23 (6.75) & 2.86 (5.66) & 1.69 (4.09) \\ 
  &F1 & 0.78 (0.16) & 0.84 (0.17) & 0.88 (0.11) & 0.94 (0.09) & 0.95 (0.08) & 0.97 (0.06) & 0.98 (0.04) & 0.99 (0.03) \\ 
  \hline
  \multirow{3}{*}{loadings 2} & FN & 16.6 (14.3) & 15.8 (16.8) & 5.72 (6.09) & 5.22 (7.72) & 2.55 (4.9) & 2.16 (5.15) & 0.61 (3.37) & 0.59 (3.4) \\ 
  &FP & 29.5 (27.1) & 24.0 (24.6) & 12.3 (18.2) & 8.74 (12.0) & 3.52 (7.45) & 3.08 (6.54) & 1.45 (4.72) & 1.40 (3.87) \\ 
  &F1& 0.79 (0.16) & 0.81 (0.17) & 0.92 (0.09) & 0.93 (0.08) & 0.97 (0.05) & 0.97 (0.05) & 0.99 (0.04) & 0.99 (0.03) \\ 
  \hline
    \multirow{3}{*}{loadings 3} &  FN & 18.9 (9.90) & 15.9 (14.1) & 9.06 (4.61) & 7.38 (4.34) & 4.39 (2.66) & 3.55 (2.41) & 0.98 (1.16) & 0.87 (1.11) \\ 
 & FP & 10.4 (13.9) & 8.87 (13.6) & 3.11 (6.16) & 2.47 (4.96) & 0.82 (2.25) & 1.15 (1.93) & 0.13 (0.53) & 1.10 (1.80) \\ 
 & F1 & 0.90 (0.07) & 0.91 (0.09) & 0.96 (0.03) & 0.96 (0.03) & 0.98 (0.01) & 0.98 (0.01) & 1.00 (0.00) & 0.99 (0.01) \\ 
   \hline
\end{tabular}
}
\end{table}

\section{Real data analysis}
\label{sec:example}
We apply the proposed method to the Stock-Watson dataset, which contains 132 time series from January 1959 to December 2003 \citep{stock2005,mccracken2016fred}. This dataset was first analyzed in \cite{stock2005} and the 132 series were divided into 14 categories: real output and income (17); employment and hours (30); real retail sales (1); manufacturing and trade sales (1); consumption (1); housing starts and sales (10); real inventories (3); orders (7); stock prices (4); exchange rates (5); interest rates and spreads (17); money and credit quantity aggregates (11); price indexes (21); average hourly earnings (3); and miscellaneous (1). The detailed variables and categories are included in Appendix A. We follow the empirical analysis in \cite{stock2005}, which suggests using 7 factors with $r=7$, and we use the 14 categories to group 132 time series with ${\cal G}_1=\cdots={\cal G}_{7}$ and $J_1=J_2=\cdots=J_7=14$. 

To make the series stationary, we follow the approach in \cite{stock2005} and transformed them by taking logarithms and/or differencing. We adjust the dataset by replacing the observations whose absolute median deviations greater than 6 times the interquartile range with the median of the preceding five observations. When estimating the factor loadings, we use the outlier-adjusted data; for the out-of-sample forecasting, we use the unadjusted/original data. Details about the transformation and outlier adjustments can be found in \cite{stock2005}.

Figure \ref{fig::zeros} shows the estimated loadings of the seven factors obtained using the method of \cite{liu::wang2025} (``sparse") and the proposed method (``sparse group") with positive loadings in red, negative loadings in purple, and zero loadings in white. Variables from different categories are separated by black lines. The method of \cite{liu::wang2025} yields 548 zero loadings, but the proposed method has 594 zero loadings with 51 group zeros and 299 individual zeros, which is much sparser. The estimated loading matrix is in Appendix A. We compute the out-of-sample prediction errors for these two methods in Table \ref{tab:prediction} and show that the estimate by the proposed method can capture the dynamics and patterns of the data very well, although it is sparser.

\begin{figure}[H]
\centering
\includegraphics[scale=0.8]{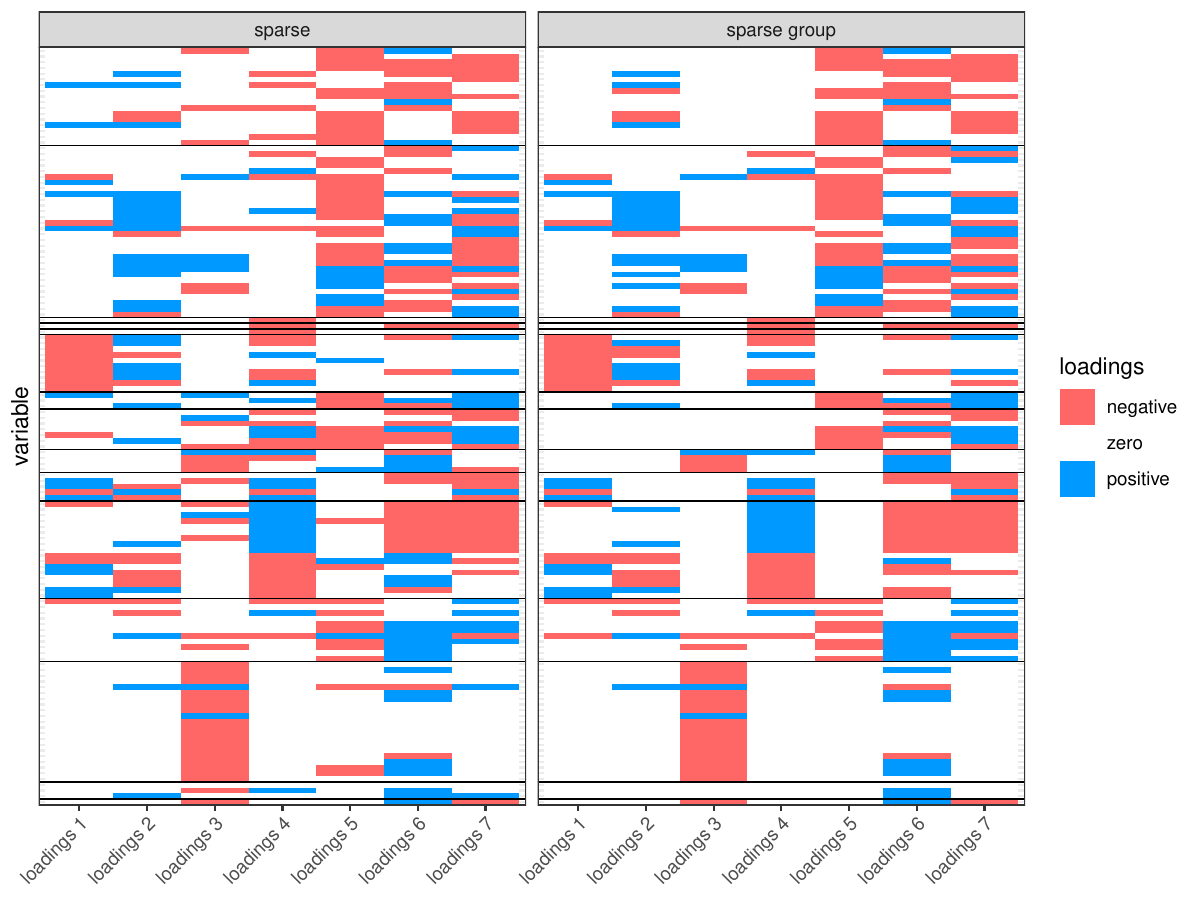}
\caption{Estimated loadings of seven factors. Left panel: Estimated sparse loadings; Right panel: Estimated sparse group loadings.}
\label{fig::zeros}
\end{figure}

We conduct one-step ahead predictions of $\bx_{t^{\prime}}$ for each $t^{\prime}=n-\tilde{n}+1, \ldots, n$. Specifically, to predict $\bx_{t^{\prime}}$, we fit the data $\bx_1$,\dots, $\bx_{t^\prime-1}$ using the sparse-group factor model to obtain the estimated loading matrix $\hat{\bQ}$ and factor processes $\hat{\bff}_t$ when $r = 7$. We apply a VAR(1) model to predict $\bff_{t^{\prime}}$, denoted as $\hat{\bff}_{t^\prime}$, and compute the prediction $\hat{\bx}_{t^{\prime}}  = \hat{\bQ}\hat{\bff}_{t^\prime}$. We repeat this rolling prediction for $t^{\prime} = n - \tilde{n} +1, \dots, n$, where $\tilde{n} = 100$. For each variable $x_{t^\prime,j}$, $j=1,\dots,p$, we compute the root mean squared error (RMSE), $\text{RMSE} = \sqrt{\tilde{n}^{-1}\sum_{t^\prime = n-\tilde{n}+1}^{n} (\hat{x}_{t^\prime,j} - x_{t^\prime,j})^2}$, and mean absolute error (MAE), $\text{MAE} =\tilde{n}^{-1}\sum_{t^\prime = n-\tilde{n}+1}^{n} \vert \hat{x}_{t^\prime,j} - x_{t^\prime,j}\vert$. Table \ref{tab:prediction} reports the average and median prediction errors across the 132 variables. Both the ``sparse" and the``sparsegroup" outperform the ``eigen" approach. The MAE of the proposed method is lower than that of the ``sparse"; while the RMSE of the ``sparsegroup" is slightly worse. These results indicate that the estimate obtained by the proposed method is sparse but can still capture the dynamics of the 132 series very well.

\begin{table}[H]
\centering
\caption{Prediction performance}
\label{tab:prediction}
\begin{tabular}{lrrrr}
  \hline
  & \multicolumn{2}{c}{RMSE} & \multicolumn{2}{c}{MAE}\\
  \hline
 method & average & median & average & median \\ 
  \hline
eigen  & 0.909 & 0.794 & 0.639 & 0.599 \\ 
 sparse & 0.906 & 0.790 & 0.637 & 0.600 \\ 
 sparsegroup &  0.906 & 0.791 & 0.636 & 0.597 \\ 
   \hline
\end{tabular}
\end{table}


Table \ref{tbl::emp_est} reports the estimated loadings of the seven factors obtained using the proposed method. Based on our estimation results, the first factor loads heavily on housing, number of employees in different sectors, and monetary market indicators (exchange rates, short-term interest rate spreads, and money supply). The second factor mainly loads on industrial production (employment, inventory, and housing starts and sales especially in Northeast and South, commodity prices index) and financial indicators (long-term interests rates and spread, money and credit quantity aggregates). The third factor loads on employment, consumer expect, prices (stock prices and price indexes), and money and credit quantity aggregates. The fourth factor is mainly associated with number of employees, housing starts and sales (U.S. West excluded), consumption and sales (retail and manufacturing and trade sales), and financial indicators (stocks' dividend yield, exchange rates, interest rates and spreads), and money and credit quantity aggregates. The fifth factor loads heavily on output, employment, inventories, orders, and money and credit quantity aggregates. The sixth factor is highly associated with output, employment (number of employees, hours, hourly earnings), manufacturing and trade sales, housing in Midwest, inventories, orders, consumer expect, and financial and monetary market indicators (stock prices, exchange rates, interest rates and spreads, money and credit quantity aggregates, and price indexes). The seventh factor loads on output, employment, manufacturing and trade sales, housing starts and sales mainly in Midwest, inventory, orders, consumer expect, and monetary market indicators (exchange rates, interest rates and one-year interest rate spreads, money and credit quantity aggregates).


\section{Conclusion}
\label{sec:conclusion}

In this paper, we extend the sparse factor model of \cite{liu::wang2025} by including both individual sparsity and group sparsity in the loading matrix. We formulate an optimization problem with penalty functions and develop an ADMM-based algorithm to estimate the loading matrix. The proposed method can incorporate the prior knowledge to obtain a more sparse and interpretable estimate. We use a simulation study to illustrate the superiority of the sparse-group estimator compared to the sparse estimator when the group information is known. We also compare different estimators using the Stock-Watson dataset with 14 predefined categories, and show the advantages of the sparse-group estimator. 


\appendix
\section{Estimation of Loading Matrix}

   \scriptsize
   \begin{center}
    \begin{longtable}{|l|l|l|r|r|r|r|r|r|r|}   
          \caption{Estimated Loading Matrix for Section \ref{sec:example}}  \label{tbl::emp_est}\\
    \hline 
        Group & name & short description &factor 1 &factor 2 &factor 3 &factor 4 &factor 5 &factor 6 &factor 7 \\ \hline 
        1 & A0M051 & PI less transfers & 0 & 0 & 0 & 0 & -0.131 & 0.019 & 0 \\ \hline
        1 & A0m082 & Cap util & 0 & 0 & 0 & 0 & -0.062 & 0 & -0.161 \\ \hline
        1 & IPS10 & IP: total & 0 & 0 & 0 & 0 & -0.108 & -0.006 & -0.097 \\ \hline
        1 & IPS11 & IP: products & 0 & 0 & 0 & 0 & -0.061 & -0.018 & -0.069 \\ \hline
        1 & IPS12 & IP: cons gds & 0 & 0.010 & 0 & 0 & 0 & -0.008 & -0.107 \\ \hline
        1 & IPS13 & IP: cons dble & 0 & 0 & 0 & 0 & 0 & 0 & -0.175 \\ \hline
        1 & IPS18 & iIP:cons nondble & 0 & 0.022 & 0 & 0 & 0 & -0.011 & 0 \\ \hline
        1 & IPS25 & IP:bus eqpt & 0 & -0.002 & 0 & 0 & -0.116 & -0.055 & 0 \\ \hline
        1 & IPS299 & IP: final prod & 0 & 0 & 0 & 0 & -0.040 & -0.028 & -0.054 \\ \hline
        1 & IPS306 & IP: fuels & 0 & 0 & 0 & 0 & 0 & 0.022 & 0 \\ \hline
        1 & IPS307 & IP: res util & 0 & 0 & 0 & 0 & 0 & -0.024 & 0 \\ \hline
        1 & IPS32 & IP: matls & 0 & -0.015 & 0 & 0 & -0.121 & 0 & -0.100 \\ \hline
        1 & IPS34 & IP: dble mats & 0 & -0.050 & 0 & 0 & -0.133 & 0 & -0.078 \\ \hline
        1 & IPS38 & IP:nondble mats & 0 & 0.028 & 0 & 0 & -0.046 & 0 & -0.177 \\ \hline
        1 & IPS43 & IP: mfg & 0 & 0 & 0 & 0 & -0.106 & 0 & -0.117 \\ \hline
        1 & PMP & NAPM prodn & 0 & 0 & 0 & 0 & -0.212 & 0 & 0 \\ \hline
        1 & a0m052 & PI & 0 & 0 & 0 & 0 & -0.10 & 0.028 & 0 \\ \hline
        2 & A0M005 & UI claims & 0 & 0 & 0 & 0 & 0 & -0.006 & 0.165 \\ \hline
        2 & A0M048 & Emp-hrs nonag & 0 & 0 & 0 & 0 & 0 & -0.135 & -0.010 \\ \hline
        2 & CES002 & Emp: total & 0 & 0 & 0 & 0 & -0.186 & 0 & 0.010 \\ \hline
        2 & CES003 & Emp: gds prod & 0 & 0 & 0 & 0 & -0.186 & 0 & 0 \\ \hline
        2 & CES006 & Emp: mining & 0 & 0 & 0 & 0.054 & 0 & -0.159 & 0 \\ \hline
        2 & CES011 & Emp: const & -0.023 & 0 & 0.015 & -0.005 & -0.019 & 0 & 0 \\ \hline
        2 & CES015 & Emp: mfg & 0.002 & 0 & 0 & 0 & -0.210 & 0 & 0 \\ \hline
        2 & CES017 & Emp: dble gds & 0 & 0 & 0 & 0 & -0.195 & 0 & 0 \\ \hline
        2 & CES033 & Emp: nondbles & 0.036 & 0.033 & 0 & 0 & -0.183 & 0.044 & -0.04 \\ \hline
        2 & CES046 & Emp: services & 0 & 0.054 & 0 & 0 & -0.145 & 0 & 0.080 \\ \hline
        2 & CES048 & Emp: TTU & 0 & 0.008 & 0 & 0 & -0.157 & 0 & 0.010 \\ \hline
        2 & CES049 & Emp: wholesale & 0 & 0.051 & 0 & 0 & -0.161 & 0 & 0.073 \\ \hline
        2 & CES053 & Emp: retail & 0 & 0.078 & 0 & 0 & -0.093 & 0.036 & 0 \\ \hline
        2 & CES088 & Emp: FIRE & -0.031 & 0.224 & 0 & 0 & 0 & 0.062 & -0.056 \\ \hline
        2 & CES140 & Emp: Govt & 0.008 & 0.148 & -0.036 & -0.098 & 0 & 0 & 0.105 \\ \hline
        2 & CES151 & Avg hrs & 0 & -0.367 & 0 & 0 & -0.303 & 0 & 0.408 \\ \hline
        2 & CES155 & Overtime: mfg & 0 & 0 & 0 & 0 & 0 & 0 & -0.127 \\ \hline
        2 & LHEL & Help wanted indx & 0 & 0 & 0 & 0 & -0.099 & 0.129 & -0.010 \\ \hline
        2 & LHELX & Help wanted/emp & 0 & 0 & 0 & 0 & -0.126 & 0.141 & 0 \\ \hline
        2 & LHEM & Emp CPS total & 0 & 0.044 & 0.03 & 0 & -0.051 & 0 & -0.004 \\ \hline
        2 & LHNAG & Emp CPS nonag & 0 & 0.057 & 0.047 & 0 & -0.079 & 0.033 & -0.008 \\ \hline
        2 & LHU14 & U 5-14 wks & 0 & 0 & 0.034 & 0 & 0.061 & -0.033 & 0.055 \\ \hline
        2 & LHU15 & U 15+ wks & 0 & 0.009 & 0 & 0 & 0.176 & -0.087 & -0.116 \\ \hline
        2 & LHU26 & U 15-26 wks & 0 & 0 & 0 & 0 & 0.119 & -0.087 & 0 \\ \hline
        2 & LHU27 & U 27+ wks & 0 & 0.011 & -0.036 & 0 & 0.131 & 0 & -0.135 \\ \hline
        2 & LHU5 & U < 5 wks & 0 & 0 & -0.033 & 0 & 0 & -0.067 & 0.024 \\ \hline
        2 & LHU680 & U: mean duration & 0 & 0 & 0 & 0 & 0.092 & 0 & -0.135 \\ \hline
        2 & LHUR & U: all & 0 & 0 & 0 & 0 & 0.141 & -0.114 & 0 \\ \hline
        2 & PMEMP & NAPM empl & 0 & 0.101 & 0 & 0 & -0.225 & -0.020 & 0.129 \\ \hline
        2 & aom001 & Avg hrs: mfg & 0 & -0.361 & 0 & 0 & -0.303 & 0 & 0.416 \\ \hline
        3 & A0M059 & Retail sales & 0 & 0 & 0 & -0.025 & 0 & 0 & 0 \\ \hline
        4 & A0M057 & M\&T sales & 0 & 0 & 0 & -0.011 & 0 & -0.009 & -0.040 \\ \hline
        5 & A0M224\_R & Consumption & 0 & 0 & 0 & -0.067 & 0 & 0 & 0 \\ \hline
        6 & HSBMW & BP: MW & -0.236 & 0 & 0 & -0.160 & 0 & -0.119 & 0.134 \\ \hline
        6 & HSBNE & BP: NE & -0.131 & 0.387 & 0 & -0.173 & 0 & 0 & 0 \\ \hline
        6 & HSBR & BP: total & -0.397 & -0.014 & 0 & 0 & 0 & 0 & 0 \\ \hline
        6 & HSBSOU & BP: South & -0.383 & -0.246 & 0 & 0.160 & 0 & 0 & 0 \\ \hline
        6 & HSBWST & BP: West & -0.379 & 0 & 0 & 0 & 0 & 0 & 0 \\ \hline
        6 & HSFR & HStarts: Total & -0.354 & 0.080 & 0 & 0 & 0 & 0 & 0 \\ \hline
        6 & HSMW & HStarts: MW & -0.182 & 0.126 & 0 & -0.171 & 0 & -0.143 & 0.108 \\ \hline
        6 & HSNE & HStarts: NE & -0.072 & 0.415 & 0 & -0.151 & 0 & 0 & 0 \\ \hline
        6 & HSSOU & HStarts: South & -0.365 & -0.045 & 0 & 0.106 & 0 & 0 & -0.005 \\ \hline
        6 & HSWST & HStarts: West & -0.359 & 0 & 0 & 0 & 0 & 0 & 0 \\ \hline
        7 & A0M070 & M\&T invent & 0 & 0 & 0 & 0 & -0.210 & 0 & 0.259 \\ \hline
        7 & A0M077 & M\&T invent/sales & 0 & 0 & 0 & 0 & -0.039 & 0.023 & 0.149 \\ \hline
        7 & PMNV & NAPM Invent & 0 & 0.170 & 0 & 0 & -0.137 & -0.081 & 0.176 \\ \hline
        8 & A0M007 & Orders: dble gds & 0 & 0 & 0 & 0 & 0 & -0.011 & -0.048 \\ \hline
        8 & A0M008 & Orders: cons gds & 0 & 0 & 0 & 0 & 0 & 0 & -0.145 \\ \hline
        8 & A0M027 & Orders: cap gds & 0 & 0 & 0 & 0 & 0 & -0.033 & 0 \\ \hline
        8 & A1M092 & Unf orders: dble & 0 & 0 & 0 & 0 & -0.175 & 0.011 & 0.103 \\ \hline
        8 & PMDEL & NAPM vendor del & 0 & 0 & 0 & 0 & -0.186 & -0.019 & 0.243 \\ \hline
        8 & PMI & PMI & 0 & 0 & 0 & 0 & -0.231 & 0 & 0.093 \\ \hline
        8 & PMNO & NAPM new ordrs & 0 & 0 & 0 & 0 & -0.191 & 0 & -0.022 \\ \hline
        9 & FSDXP & S\&P div yield & 0 & 0 & 0.085 & 0.010 & 0 & -0.239 & 0 \\ \hline
        9 & FSPCOM & S\&P 500 & 0 & 0 & -0.101 & 0 & 0 & 0.204 & 0 \\ \hline
        9 & FSPIN & S\&P: indust & 0 & 0 & -0.097 & 0 & 0 & 0.207 & 0 \\ \hline
        9 & FSPXE & S\&P PE ratio & 0 & 0 & -0.038 & 0 & 0 & 0.189 & 0 \\ \hline
        10 & EXRCAN & EX rate: Canada & 0 & 0 & 0 & 0 & 0 & -0.046 & -0.039 \\ \hline
        10 & EXRJAN & Ex rate: Japan & 0.051 & 0 & 0 & 0.083 & 0 & -0.008 & -0.112 \\ \hline
        10 & EXRSW & Ex rate: Switz & 0.066 & 0 & 0 & 0.124 & 0 & 0 & -0.198 \\ \hline
        10 & EXRUK & Ex rate: UK & -0.057 & 0 & 0 & -0.071 & 0 & 0 & 0.131 \\ \hline
        10 & EXRUS & Ex rate: avg & 0.090 & 0 & 0 & 0.129 & 0 & 0 & -0.183 \\ \hline
        11 & CP90 & Commpaper & -0.006 & 0 & 0 & 0.051 & 0 & -0.195 & -0.149 \\ \hline
        11 & FYAAAC & Aaabond & 0 & 0.007 & 0 & 0.115 & 0 & -0.184 & -0.203 \\ \hline
        11 & FYBAAC & Baa bond & 0 & 0 & 0 & 0.128 & 0 & -0.259 & -0.185 \\ \hline
        11 & FYFF & FedFunds & 0 & 0 & 0 & 0.013 & 0 & -0.191 & -0.103 \\ \hline
        11 & FYGM3 & 3 mo T-bill & 0 & 0 & 0 & 0.093 & 0 & -0.131 & -0.182 \\ \hline
        11 & FYGM6 & 6 mo T-bill & 0 & 0 & 0 & 0.115 & 0 & -0.132 & -0.202 \\ \hline
        11 & FYGT1 & 1 yr T-bond & 0 & 0 & 0 & 0.133 & 0 & -0.164 & -0.209 \\ \hline
        11 & FYGT10 & 10 yr T-bond & 0 & 0.022 & 0 & 0.123 & 0 & -0.066 & -0.215 \\ \hline
        11 & FYGT5 & 5 yr T-bond & 0 & 0 & 0 & 0.129 & 0 & -0.111 & -0.239 \\ \hline
        11 & sFYAAAC & Aaa-FF spread & -0.069 & -0.189 & 0 & -0.268 & 0 & 0 & 0 \\ \hline
        11 & sFYBAAC & Baa-FF spread & -0.077 & -0.175 & 0 & -0.223 & 0 & 0.032 & 0 \\ \hline
        11 & sFYGM6 & 6 mo-FF spread & 0.038 & 0 & 0 & -0.332 & 0 & -0.044 & 0 \\ \hline
        11 & sFYGT1 & 1 yr-FF spread & 0.046 & -0.053 & 0 & -0.242 & 0 & -0.046 & -0.112 \\ \hline
        11 & sFYGT10 & 10yr-FF spread & 0 & -0.138 & 0 & -0.281 & 0 & 0 & 0 \\ \hline
        11 & sFYGT5 & 5 yr-FFspread & 0 & -0.123 & 0 & -0.283 & 0 & 0 & 0 \\ \hline
        11 & scp90 & CP-FF spread & 0.112 & 0.080 & 0 & -0.330 & 0 & -0.156 & 0 \\ \hline
        11 & sfygm3 & 3 mo-FF spread & 0.006 & 0 & 0 & -0.326 & 0 & -0.010 & 0 \\ \hline
        12 & A0M095 & Inst cred/PI & -0.078 & -0.058 & 0 & -0.027 & -0.009 & 0 & 0.083 \\ \hline
        12 & CCINRV & Cons credit & 0 & 0 & 0 & 0 & 0 & 0 & 0 \\ \hline
        12 & FCLBMC & C\&I loans & 0 & -0.155 & 0 & 0.143 & -0.145 & 0 & 0.114 \\ \hline
        12 & FCLNQ & C\&I loans & 0 & 0 & 0 & 0 & 0 & 0 & 0 \\ \hline
        12 & FM1 & M1 & 0 & 0 & 0 & 0 & -0.047 & 0.248 & 0.127 \\ \hline
        12 & FM2 & M2 & 0 & 0 & 0 & 0 & -0.042 & 0.257 & 0.161 \\ \hline
        12 & FM2DQ & M2 (real) & -0.010 & 0.139 & -0.053 & -0.116 & 0 & 0.264 & -0.033 \\ \hline
        12 & FM3 & M3 & 0 & 0 & 0 & 0 & -0.054 & 0.189 & 0.092 \\ \hline
        12 & FMFBA & MB & 0 & 0 & -0.017 & 0 & -0.018 & 0.180 & 0.024 \\ \hline
        12 & FMRNBA & Reserves nonbor & 0 & 0 & 0 & 0 & 0 & 0.034 & 0 \\ \hline
        12 & FMRRA & Reserves tot & 0 & 0 & 0 & 0 & -0.027 & 0.206 & 0.025 \\ \hline
        13 & GMDC & PCE defl & 0 & 0 & -0.268 & 0 & 0 & 0 & 0 \\ \hline
        13 & GMDCD & PCE defl: dlbes & 0 & 0 & -0.064 & 0 & 0 & 0.042 & 0 \\ \hline
        13 & GMDCN & PCE defl: nondble & 0 & 0 & -0.309 & 0 & 0 & 0 & 0 \\ \hline
        13 & GMDCS & PCE defl: services & 0 & 0 & -0.125 & 0 & 0 & 0 & 0 \\ \hline
        13 & PMCP & NAPM com price & 0 & 0.191 & 0.027 & 0 & 0 & -0.180 & 0 \\ \hline
        13 & PSCCOM & Commod: spot price & 0 & 0 & -0.038 & 0 & 0 & 0.068 & 0 \\ \hline
        13 & PSM99Q & Sens mat’ls price & 0 & 0 & -0.021 & 0 & 0 & 0.138 & 0 \\ \hline
        13 & PU83 & CPI-U: apparel & 0 & 0 & -0.048 & 0 & 0 & 0 & 0 \\ \hline
        13 & PU84 & CPI-U: transp & 0 & 0 & -0.208 & 0 & 0 & 0 & 0 \\ \hline
        13 & PU85 & CPI-U: medical & 0 & 0 & 0.073 & 0 & 0 & 0 & 0 \\ \hline
        13 & PUC & CPI-U: comm. & 0 & 0 & -0.366 & 0 & 0 & 0 & 0 \\ \hline
        13 & PUCD & CPI-U: dbles & 0 & 0 & -0.103 & 0 & 0 & 0 & 0 \\ \hline
        13 & PUNEW & CPI-U: all & 0 & 0 & -0.342 & 0 & 0 & 0 & 0 \\ \hline
        13 & PUS & CPI-U: services & 0 & 0 & -0.077 & 0 & 0 & 0 & 0 \\ \hline
        13 & PUXF & CPI-U: ex food & 0 & 0 & -0.263 & 0 & 0 & 0 & 0 \\ \hline
        13 & PUXHS & CPI-U: ex shelter & 0 & 0 & -0.328 & 0 & 0 & 0 & 0 \\ \hline
        13 & PUXM & CPI-U: ex med & 0 & 0 & -0.352 & 0 & 0 & -0.008 & 0 \\ \hline
        13 & PWCMSA & PPI: crude mat’ls & 0 & 0 & -0.173 & 0 & 0 & 0.007 & 0 \\ \hline
        13 & PWFCSA & PPI: cons gds & 0 & 0 & -0.216 & 0 & 0 & 0.107 & 0 \\ \hline
        13 & PWFSA & PPI: fin gds & 0 & 0 & -0.188 & 0 & 0 & 0.111 & 0 \\ \hline
        13 & PWIMSA & PPI: int mat’ls & 0 & 0 & -0.167 & 0 & 0 & 0 & 0 \\ \hline
        14 & CES275 & AHE: goods & 0 & 0 & 0 & 0 & 0 & 0 & 0 \\ \hline
        14 & CES277 & AHE: const & 0 & 0 & 0 & 0 & 0 & 0.057 & 0 \\ \hline
        14 & CES278 & AHE: mfg & 0 & 0 & 0 & 0 & 0 & 0.007 & 0 \\ \hline
    \end{longtable}
\end{center}
\vspace{-1.5cm}
\normalsize
The variable description can be found in Appendix A in \cite{stock2005}.

\section{Proof of Theorem 1}

\begin{proof}

We will prove the results in two parts. In part 1, we prove the result for $\hat{\bf q}_1$. 

\textbf{Part 1}: $\Vert  \hat{\bf q}_1 - {\bf q}_1 \Vert_2 = O_p(\tau_{n,p,m})$.

From Lemma \ref{lem_or}, we know that $\Vert\hat{{\bf q}}_{1}^{or}-{\bf q}_{1}\Vert_{2}=O_{p}\left(\tau_{n,p,m}\right)$. The next step is to show that $\hat{{\bf q}}_{1}^{or}$ is a local minimizer
of $$G\left({\bf q}_{1}\right)=\frac{1}{2}\Vert\hat{{\bf S}}\hat{{\bf S}}^{\top}-{\bf q}_{1}{\bf q}_{1}^{\top}\Vert_{F}^{2}+\sum_{j=1}^{p}\mathcal{P}_{\gamma}\left(\vert q_{1j}\vert,\lambda_1\right)+\sum_{j^{\prime}=1}^{J_{1}}\mathcal{P}_{\gamma}\left(\Vert{\bf q}_{1\left(j^{\prime}\right)}\Vert_2,\sqrt{d_{1j^{\prime}}}\lambda_2\right),$$
subject to $\Vert{\bf q}_{1}\Vert_2=1$.

Consider a neighbor of ${\bf q}_{1}$ such that $\Vert{\bf u}-{\bf q}_{1}\Vert_2=O_{p}\left(\tau_{n,p,m}\right)$ and $\Vert{\bf u}\Vert_2=1$.  Define ${\bf u}^{*}[\mathcal{V}_{1}]={\bf u}_{[\mathcal{V}_{1}]}$
and ${\bf u}^{*}[-\mathcal{V}_{1}]=\bzero$, and $\alpha=\Vert{\bf u}^{*}\Vert_2$.
Let $\tilde{{\bf u}}={\bf u}^{*}/\alpha$ , which indicates that $\tilde{{\bf u}}[-\mathcal{V}_{1}]={\bf 0}$ and $\Vert \tilde{\bu}\Vert_2 =1$ based on the definition of $\tilde{\bu}$.

First we will compare $G\left(\hat{{\bf q}}_{1}^{or}\right)$ and
$G\left(\tilde{{\bf u}}\right)$. We have $\Vert{\bf u}^{*}-{\bf q}_{1}\Vert_2=O_{p}\left(\tau_{n,p,m}\right)$,
$\Vert{\bf u}_{[-\mathcal{V}_{1}]}\Vert_2=\Vert{\bf u}-{\bf u}^{*}\Vert_2=O_{p}\left(\tau_{n,p,m}\right)$,
and $\alpha=\Vert{\bf u}^{*}\Vert_2\geq\Vert{\bf q}_{1}\Vert_2 - \Vert{\bf u}^{*}-{\bf q}_{1}\Vert_2=1-\Vert{\bf u}^{*}-{\bf q}_{1}\Vert_2$ with $\alpha <1$. We can have 
\begin{align*}
\tilde{{\bf u}}-{\bf q}_{1} & =\frac{{\bf u}^{*}}{\alpha}-{\bf q}_{1}=\frac{{\bf u}-\left({\bf u}-{\bf u}^{*}\right)}{\alpha}-{\bf q}_{1} =\frac{{\bf u-{\bf q}_{1}-\left({\bf u}-{\bf u}^{*}\right)}}{\alpha}+\left(\frac{1}{\alpha}-1\right){\bf q}_{1}.
\end{align*}
Thus 
\begin{align*}
\Vert\tilde{{\bf u}}-{\bf q}_{1}\Vert_2 & \leq\frac{1}{\alpha}\Vert{\bf u}-{\bf q}_{1}\Vert_2+\frac{1}{\alpha}\Vert{\bf u}-{\bf u}^{*}\Vert_2+\frac{1}{\alpha}-1\\
&= \frac{\Vert{\bf u}-{\bf q}_{1}\Vert_2+\Vert{\bf u}-{\bf u}^{*}\Vert_2+\Vert{\bf u}-{\bf q}_{1}\Vert_2}{1-\Vert \bu^* - \bq_1 \Vert_2 } = O_p(\tau_{n,p,m}).
\end{align*}

Based on the assumption about the minimal signal and the assumption about $\lambda_1$ and $\lambda_2$, for $j\in\mathcal{V}_{1}$, we have
$\vert\hat{q}_{1j}^{or}\vert\geq\vert q_{1j}\vert-\vert q_{1j}-\hat{q}_{1j}^{or}\vert>\gamma\lambda_1$
since $\vert q_{1j}-\hat{q}_{1j}^{or}\vert=O_{p}\left(\tau_{n,p,m}\right)$. Similarly we have $\vert\tilde{u}_{j}\vert>\gamma\lambda_1$. Then, $\mathcal{P}_{\gamma}\left(\vert\hat{q}_{1j}^{or}\vert,\lambda_1\right)=\mathcal{P}_{\gamma}\left(\vert\tilde{u}_{j}\vert,\lambda_1\right)=\frac{1}{2}\gamma\lambda_1^{2}$ based on the definition MCP.
So we have $\sum_{j=1}^{p}\mathcal{P}_{\gamma}\left(\vert\hat{q}_{1j}^{or}\vert,\lambda_1\right)=\sum_{j=1}^{p}\mathcal{P}_{\gamma}\left(\vert\tilde{u}_{j}\vert,\lambda_1\right)$.

Furthermore, for $j^\prime \in \mathcal{V}^g_1$, we have $\Vert \hat{{\bf q}}_{1(j^\prime)}^{or}\Vert_2 \geq \Vert {\bf q}_{1(j^\prime)}\Vert_2 - \Vert {\bf q}_{1(j^\prime)} - \hat{{\bf q}}_{1(j^\prime)}^{or} \Vert_2 > \sqrt{d_{1j^\prime}}\gamma \lambda_2$  since $\Vert {\bf q}_{1(j^\prime)} - \hat{{\bf q}}_{1(j^\prime)}^{or} \Vert_2 = O_p(\tau_{n,p,m})$. We can have $\Vert \tilde{{\bf u}}_{(j^\prime)}\Vert_2 > \sqrt{d_{1j^\prime}}\gamma \lambda_2$ since $\Vert\tilde{{\bf u}}-{\bf q}_{1}\Vert_2 = O_p(\tau_{n,p,m})$. Thus, $\sum_{j^{\prime}=1}^{J_{1}}\mathcal{P}_{\gamma}\left(\Vert{\tilde{\bf u}}_{\left(j^{\prime}\right)}\Vert,\sqrt{d_{1j^{\prime}}}\lambda_2\right) =\sum_{j^{\prime}=1}^{J_{1}}\mathcal{P}_{\gamma}\left(\Vert{\bf q}_{1\left(j^{\prime}\right)}^{or}\Vert,\sqrt{d_{1j^{\prime}}}\lambda_2\right) = \frac{1}{2}\sum_{j^\prime}^{J_1}d_{1j^\prime}\gamma \lambda_2^2$ based on the definition of MCP. 

Based on the definition of $\hat{{\bf q}}_{1}^{or}$, we have $\Vert\hat{{\bf S}}\hat{{\bf S}}^{\top}-\hat{{\bf q}}_{1}^{or}\left(\hat{{\bf q}}_{1}^{or}\right)^{\top}\Vert_{F}^{2}<\Vert\hat{{\bf S}}\hat{{\bf S}}^{\top}-\tilde{{\bf u}}\tilde{{\bf u}}^{\top}\Vert_{F}^{2}$
for $\tilde{{\bf u}}\neq\hat{{\bf q}}_{1}^{or}$. This implies that
$G\left(\hat{{\bf q}}_{1}^{or}\right)<G\left(\tilde{{\bf u}}\right)$. 

Next, we will compare $G\left(\tilde{{\bf u}}\right)$ and $G\left({\bf u}\right)$.
We have that 
\begin{align}
\label{eq_gdiff}
G\left(\tilde{{\bf u}}\right)-G\left({\bf u}\right) & =-\tilde{{\bf u}}^{\top}\hat{{\bf S}}\hat{{\bf S}}^{\top}\tilde{{\bf u}}+{\bf u}^{\top}\hat{{\bf S}}\hat{{\bf S}}^{\top}{\bf u} \nonumber\\
 & +\sum_{j=1}^{p}\mathcal{P}_{\gamma}\left(\vert\tilde{u}_{j}\vert,\lambda_{1}\right)-\sum_{j=1}^{p}\mathcal{P}_{\gamma}\left(\vert u_{j}\vert,\lambda_1\right)+ \nonumber\\
 & +\sum_{j^{\prime}=1}^{J_{1}}\mathcal{P}_{\gamma}\left(\Vert\tilde{{\bf u}}_{(j^{\prime})}\Vert_2,\sqrt{d_{1j^{\prime}}}\lambda_{2}\right)-\sum_{j^{\prime}=1}^{J_{1}}\mathcal{P}_{\gamma}\left(\Vert{\bf u}_{(j^{\prime})}\Vert_2,\sqrt{d_{1j^{\prime}}}\lambda_{2}\right)
\end{align}

Similar to the proof in \cite{liu::wang2025}, we first consider $-\tilde{{\bf u}}^{\top}\hat{{\bf S}}\hat{{\bf S}}^{\top}\tilde{{\bf u}}+{\bf u}^{\top}\hat{{\bf S}}\hat{{\bf S}}^{\top}{\bf u}$ in \eqref{eq_gdiff}.

The following part is the same as that of in \cite{liu::wang2025}.
Since $-\tilde{{\bf u}}^{\top}\hat{{\bf S}}\hat{{\bf S}}^{\top}\tilde{{\bf u}}=-\frac{1}{\alpha^{2}}{\bf u}^{*^{\top}}\hat{{\bf S}}\hat{{\bf S}}^{\top}{\bf u}^{*}\leq-{\bf u}^{*^{\top}}\hat{{\bf S}}\hat{{\bf S}}^{\top}{\bf u}^{*}$. Thus, \begin{align}
\label{eq_part1}
-\tilde{{\bf u}}^{\top}\hat{{\bf S}}\hat{{\bf S}}^{\top}\tilde{{\bf u}}+{\bf u}^{\top}\hat{{\bf S}}\hat{{\bf S}}^{\top}{\bf u} & \leq-{\bf u}^{*^{\top}}\hat{{\bf S}}\hat{{\bf S}}^{\top}{\bf u}^{*}+{\bf u}^{\top}\hat{{\bf S}}\hat{{\bf S}}^{\top}{\bf u} \nonumber \\
 & = {\bf u}^{*\top}\hat{{\bf S}}\hat{{\bf S}}^{\top}\left({\bf u}-{\bf u}^{*}\right)+\left({\bf u}-{\bf u}^{*}\right)^{\top}\hat{{\bf S}}\hat{{\bf S}}^{\top}{\bf {\bf u}} \nonumber\\
 & \leq\Vert{\bf u}^{*\top}\hat{{\bf S}}\hat{{\bf S}}^{\top}\left({\bf u}-{\bf u}^{*}\right)\Vert_2+\Vert\left({\bf u}-{\bf u}^{*}\right)^{\top}\hat{{\bf S}}\hat{{\bf S}}^{\top}{\bf u}\Vert_2.
\end{align}
Denote ${\bf E}=\hat{{\bf S}}\hat{{\bf S}}^{\top}-{\bf S}{\bf S}^{\top}$, we have
\begin{align}
\label{eq_part1_p1}
\Vert{\bf u}^{*\top}\hat{{\bf S}}\hat{{\bf S}}^{\top}\left({\bf u}-{\bf u}^{*}\right)\Vert_2 & =\Vert{\bf u}^{*\top}\left({\bf S}{\bf S}^{\top}+{\bf E}\right)\left({\bf u}-{\bf u}^{*}\right)\Vert_2 \nonumber\\
 & \leq\Vert{\bf u}^{*\top}{\bf S}{\bf S}^{\top}\left({\bf u}-{\bf u}^{*}\right)\Vert_2+\Vert{\bf u}^{*\top}{\bf E}\left({\bf u}-{\bf u}^{*}\right)\Vert_2.
\end{align}
Let ${\bf u}^{*}={\bf u}^*- {\bf q}_1 + {\bf q}_1={\bf q}_{1}+{\bf e}^*$, where  ${\bf e}^*={\bf u}-{\bf q}_{1}$, thus, the first part in \eqref{eq_part1_p1} can be bounded by 
\begin{align*}
\Vert{\bf u}^{*\top}{\bf S}{\bf S}^{\top}\left({\bf u}-{\bf u}^{*}\right)\Vert_2 & \leq\Vert{\bf q}_{1}^{\top}{\bf S}{\bf S}^{\top}\left({\bf u}-{\bf u}^{*}\right)\Vert_2+\Vert{\bf e}^{*\top}{\bf S}{\bf S}^{\top}\left({\bf u}-{\bf u}^{*}\right)\Vert_2\\
 & \leq0+O_{p}\left(\tau_{n,p,m} \right)\Vert{\bf u}-{\bf u}^{*}\Vert_2. 
\end{align*}
Furthermore $\Vert{\bf u}^{*\top}{\bf E}\left({\bf u}-{\bf u}^{*}\right)\Vert_2= \Vert{\bf u}^{*\top}{\bf E}\Vert_{\max} \sum_{j\notin \mathcal{V}_1}\vert u_j\vert$. From Lemma \ref{lemma_uE}, $\Vert{\bf u}^{*\top}{\bf E}\left({\bf u}-{\bf u}^{*}\right)\Vert_2= O_p(\tau_{n,p,m}) \sum_{j\notin \mathcal{V}_1}\vert u_j\vert$, which gives the bound of the second part in \eqref{eq_part1_p1}.
 Thus, the first part in \eqref{eq_part1} can be bounded as follows
$\Vert{\bf u}^{*\top}\hat{{\bf S}}\hat{{\bf S}}^{\top}\left({\bf u}-{\bf u}^{*}\right)\Vert_2=O_p(\tau_{n,p,m}) \sum_{j\notin \mathcal{V}_1}\vert u_j\vert$.

Similarly, ${\bf u}={\bf u}-{\bf q}_{1}+{\bf q}_{1}={\bf e}+{\bf q}_{1}$
, the second part in \eqref{eq_part1} can be bounded by 
\begin{align*}
\Vert\left({\bf u}-{\bf u}^{*}\right)^{\top}\hat{{\bf S}}\hat{{\bf S}}^{\top}{\bf u}\Vert_2 & =\Vert{\bf u}^{\top}\left({\bf S}{\bf S}^{\top}+{\bf E}\right)\left({\bf u}-{\bf u}^{*}\right)\Vert_2\\
 & \leq\Vert{\bf u}^{\top}{\bf S}{\bf S}^{\top}\left({\bf u}-{\bf u}^{*}\right)\Vert_2+\Vert{\bf u}^{\top}{\bf E}\left({\bf u}-{\bf u}^{*}\right)\Vert_2\\
 & \leq\Vert{\bf e}^{\top}{\bf S}{\bf S}^{\top}\left({\bf u}-{\bf u}^{*}\right)\Vert_2+\Vert{\bf q}_{1}^{\top}{\bf S}{\bf S}^{\top}\left({\bf u}-{\bf u}^{*}\right)\Vert_2\\
 & +\Vert{\bf e}^{\top}{\bf E}\left({\bf u}-{\bf u}^{*}\right)\Vert_2+\Vert{\bf q}_{1}^{\top}{\bf E}\left({\bf u}-{\bf u}^{*}\right)\Vert_2 \leq O_{p}\left(\tau_{n,p,m}\right)\sum_{j\notin\mathcal{V}_{1}}\vert u_{j}\vert.
\end{align*}

Thus, we have the following bound for  \eqref{eq_part1},
\begin{equation}
\label{eq_bound1}
-\tilde{{\bf u}}^{\top}\hat{{\bf S}}\hat{{\bf S}}^{\top}\tilde{{\bf u}}+{\bf u}^{\top}\hat{{\bf S}}\hat{{\bf S}}^{\top}{\bf u}\leq O_{p}\left(\tau_{n,p,m}\right)\sum_{j\notin\mathcal{V}_{1}}\vert u_{j}\vert.
\end{equation}

Next, consider $\sum_{j=1}^{p}\mathcal{P}_{\gamma}\left(\vert\tilde{u}_{j}\vert,\lambda_{1}\right)-\sum_{j=1}^{p}\mathcal{P}_{\gamma}\left(\vert u_{j}\vert,\lambda_1\right)$ in  \eqref{eq_gdiff}. 
Since $\Vert{\bf u}-{\bf q}_{1}\Vert_2=O_{p}\left(\tau_{n,p,m}\right)$
and $\Vert\tilde{{\bf u}}-{\bf q}_{1}\Vert_2=O_{p}\left(\tau_{n,p,m}\right)$
and $\lambda_1/\tau_{n,p,m} \to \infty$ as assumed, thus, $\vert\tilde{u}_{j}\vert>\gamma\lambda_1$
and $\vert u_{j}\vert>\gamma\lambda_1$ for $j\in\mathcal{V}_{1}$. Thus we have, 
\[
\sum_{j=1}^{p}\mathcal{P}_{\gamma}\left(\vert\tilde{u}_{j}\vert,\lambda_1\right)-\sum_{j=1}^{p}\mathcal{P}_{\gamma}\left(\vert u_{j}\vert,\lambda_1\right)=-\sum_{j\notin\mathcal{V}_{1}}\mathcal{P}_{\gamma}\left(\vert u_{j}\vert,\lambda_1\right)=-\sum_{j\notin\mathcal{V}_{1}}\left(\lambda_1\vert u_{j}\vert-\frac{\vert u_{j}\vert^{2}}{2\gamma}\right).
\]

Furthermore, since $\lambda_{2}/\tau_{n,p,m}\rightarrow\infty$, and $\Vert{\bf u}_{\left(j^{\prime}\right)}\Vert=O_{p}\left(\tau_{n,p,m}\right)$,
we have \begin{align*}
 & \sum_{j^{\prime}=1}^{J_{1}}\mathcal{P}_{\gamma}\left(\Vert\tilde{{\bf u}}_{(j^{\prime})}\Vert_2,\sqrt{d_{1j^{\prime}}}\lambda_{2}\right)-\sum_{j^{\prime}=1}^{J_{1}}\mathcal{P}_{\gamma}\left(\Vert{\bf u}_{(j^{\prime})}\Vert_2,\sqrt{d_{1j^{\prime}}}\lambda_{2}\right)\\
= & -\sum_{j^{\prime}\notin\mathcal{V}_{1}^{g}}\mathcal{P}_{\gamma}\left(\Vert{\bf u}_{(j^{\prime})}\Vert_2,\sqrt{d_{1j^{\prime}}}\lambda_{2}\right)=-\sum_{j^{\prime}\notin\mathcal{V}_{1}^{g}}\left(\sqrt{d_{1j^{\prime}}}\lambda_{2}\Vert{\bf u}_{(j^{\prime})}\Vert_2-\frac{\Vert{\bf u}_{(j^{\prime})}\Vert_2^{2}}{2\gamma}\right)\\
= & -\sum_{j^{\prime}\notin\mathcal{V}_{1}^{g}}\left(\sqrt{d_{1j^{\prime}}}\lambda_{2}-\frac{\Vert{\bf u}_{(j^{\prime})}\Vert_2}{2\gamma}\right)\Vert{\bf u}_{(j^{\prime})}\Vert_2.
\end{align*}

Thus, \eqref{eq_gdiff} is bounded by
\begin{align*}
 & G\left(\tilde{{\bf u}}\right)-G\left({\bf u}\right)\\
\leq & -\sum_{j\notin\mathcal{V}_{1}}\left(\lambda_{1}-\frac{\vert u_{j}\vert}{2\gamma}-O_{p}\left(\tau_{n,p,m}\right)\right)\vert u_{j}\vert-\sum_{j^{\prime}\notin\mathcal{V}_{1}^{g}}\left(\sqrt{d_{1j^{\prime}}}\lambda_{2}-\frac{\Vert{\bf u}_{(j^{\prime})}\Vert_2}{2\gamma}\right)\Vert{\bf u}_{(j^{\prime})}\Vert_2.
\end{align*}
As $\vert u_{j}\vert=O_{p}\left(\tau_{n,p,m}\right)$ for
$j\notin\mathcal{\mathcal{V}}_{1}$, thus $\lambda_1 \gtrsim \vert u_{j}\vert$ and $\sqrt{d_{1j^{\prime}}}\lambda_{2} \gtrsim  \Vert{\bf u}_{(j^{\prime})}\Vert_2$. 
This implies that $G\left(\tilde{{\bf u}}\right)-G\left({\bf u}\right)<0$ for $\tilde{\bf u} \neq {\bf u}$.
Thus, we have shown that $G\left(\hat{{\bf q}}_{1}^{or}\right)<G\left(\tilde{{\bf u}}\right)<G\left({\bf u}\right)$ for $\bu \neq \hat{\bq}_1^{or}$,
which implies that $\hat{{\bf q}}_{1}^{or}$ is a local minimizer
of the objective function. This completes the proof.

\textbf{Part 2}: $\Vert  \hat{\bf q}_i - {\bf q}_i \Vert_2 = O_p(\tau_{n,p,m})$ for $i\geq 2$.

From Lemma \ref{lem_or}, we know that $\Vert  \hat{\bf q}^{or}_i - {\bf q}_i \Vert_2 = O_p(\tau_{n,p,m})$.

Then, we will show that  $\hat{{\bf q}}_{i}^{or}$ is a local minimizer
of $G\left({\bf q}_{i}\right)$ with the following form,
\begin{align*}
G\left({\bf q}_{i}\right)=\frac{1}{2}\Vert\widehat{{\bf S}}\widehat{{\bf S}}^{\top}-{\bf s}_{i}{\bf s}_{i}^{\top}\Vert_{F}^{2}+\sum_{j=1}^{p}\mathcal{P}_{\gamma}\left(\vert q_{ij}\vert,\lambda_{1}\right)+\sum_{j^{\prime}=1}^{J_{i}}\mathcal{P}_{\gamma}\left(\Vert{\bf q}_{i\left(j^{\prime}\right)}\Vert_2,\sqrt{d_{ij^{\prime}}}\lambda_{2}\right)
\\
\text{subject to }{\bf s}_{i}=\left({\bf I}-\tilde{{\bf S}}_{i}\tilde{{\bf S}}_{i}^{\top}\right){\bf q}_{i}\text{ and }\Vert{\bf s}_{i}\Vert_2=1.
\end{align*}

Consider a neighbor of ${\bf q}_{i}$ such that $\Vert{\bf u}-{\bf q}_{i}\Vert_2=O_{p}\left(\tau_{n,p,m}\right)$, $\Vert{\bf u}-\hat{{\bf q}}_{i}^{or}\Vert_2 \leq \delta_n$, where $\delta_n = o(1)$,
and satisfies $\Vert\left({\bf I}-\tilde{{\bf S}}_{i}\tilde{{\bf S}}_{i}^{\top}\right){\bf u}\Vert_2=1$.
Define ${\bf u}^{*}[\mathcal{V}_{i}]={\bf u}_{[\mathcal{V}_{i}]}$
and ${\bf u}^{*}[-\mathcal{V}_{i}]=\bf{0}$, and $\alpha=\Vert\left({\bf I}-\tilde{{\bf S}}_{i}\tilde{{\bf S}}_{i}^{\top}\right){\bf u}^{*}\Vert_2$. Denote
$\tilde{{\bf u}}={\bf u}^{*}/\alpha$ , which indicates that $\tilde{{\bf u}}[-\mathcal{V}_{i}]={\bf 0}$ and $\Vert\left({\bf I}-\tilde{{\bf S}}_{i}\tilde{{\bf S}}_{i}^{\top}\right)\tilde{{\bf u}}\Vert_2=1$ based on the definition of $\tilde{\bu}$.

First we will compare $G\left(\hat{{\bf q}}_{i}^{or}\right)$ and
$G\left(\tilde{{\bf u}}\right)$. The proof is almost the same as that in the proof of Theorem 2 in \cite{liu::wang2025}.

From the definitions, we have $\Vert{\bf u}^{*}-{\bf q}_{i}\Vert_2=O_{p}\left(\tau_{n,p,m}\right)$,
$\Vert{\bf u}_{[-\mathcal{V}_{i}]}\Vert_2=\Vert{\bf u}-{\bf u}^{*}\Vert_2=O_{p}\left(\tau_{n,p,m}\right)$,
and 
$
\alpha=\Vert\left({\bf I}-\tilde{{\bf S}}_{i}\tilde{{\bf S}}_{i}^{\top}\right){\bf u}^{*}\Vert_2=\Vert\left({\bf I}-\tilde{{\bf S}}_{i}\tilde{{\bf S}}_{i}^{\top}\right){\bf u}+\left({\bf I}-\tilde{{\bf S}}_{i}\tilde{{\bf S}}_{i}^{\top}\right){\bf {\bf u}}_{[-\mathcal{V}_{i}]}\Vert_2 \geq 1-\Vert\left({\bf I}-\tilde{{\bf S}}_{i}\tilde{{\bf S}}_{i}^{\top}\right){\bf {\bf u}}_{[-\mathcal{V}_{i}]}\Vert_2
$ with $\alpha\leq1$.  We have 
\begin{align*}
\tilde{{\bf u}}-{\bf q}_{i} & =\frac{{\bf u}^{*}}{\alpha}-{\bf q}_{i}=\frac{{\bf u}-\left({\bf u}-{\bf u}^{*}\right)}{\alpha}-{\bf q}_{i}=\frac{{\bf u}-{\bf q}_{i}+\left({\bf u}-{\bf u}^{*}\right)}{\alpha}+\left(\frac{1}{\alpha}-1\right){\bf q}_{i}.
\end{align*}
Thus, 
\begin{align*}
\Vert\tilde{{\bf u}}-{\bf q}_{i}\Vert_2 & \leq\frac{1}{\alpha}\Vert{\bf u}-{\bf q}_{i}\Vert_2+\frac{1}{\alpha}\Vert{\bf u}-{\bf u}^{*}\Vert_2+\frac{1}{\alpha}-1\\
 & \leq\frac{\Vert{\bf u}-{\bf q}_{i}\Vert_2+\Vert{\bf u}-{\bf u}^{*}\Vert_2+\Vert\left({\bf I}-\tilde{{\bf S}}_{i}\tilde{{\bf S}}_{i}^{\top}\right){\bf {\bf u}}_{[-\mathcal{V}_{i}]}\Vert_2}{1-\Vert\left({\bf I}-\tilde{{\bf S}}_{i}\tilde{{\bf S}}_{i}^{\top}\right){\bf {\bf u}}_{[-\mathcal{V}_{i}]}\Vert_2} =O_{p}\left(\tau_{n,p,m}\right).
\end{align*}

Based on the assumption about the minimal signal and the assumption about $\lambda_1$, for $j\in\mathcal{V}_{i}$, we have
$\vert\hat{q}_{ij}^{or}\vert\geq\vert q_{ij}\vert- \vert q_{ij} -\hat{q}_{ij}^{or}\vert >\gamma\lambda_1$,   since $\vert q_{ij} -\hat{q}_{ij}^{or}\vert = O_p(\tau_{n,p,m})$. Similarly $\vert\tilde{u}_{j}\vert>\gamma\lambda_1$ for $j\in \mathcal{V}_i$.
Then $\mathcal{P}_{\gamma}\left(\vert\hat{q}_{ij}\vert,\lambda_1\right)=\mathcal{P}_{\gamma}\left(\vert\tilde{u}_{j}\vert,\lambda_1\right)=\frac{1}{2}\gamma\lambda_1^{2}$.
So we have $\sum_{j=1}^{p}\mathcal{P}_{\gamma}\left(\vert\hat{q}_{ij}^{or}\vert,\lambda_1\right)=\sum_{j=1}^{p}\mathcal{P}_{\gamma}\left(\vert\tilde{u}_{j}\vert,\lambda_1\right)$. Similarly, we also have $\sum_{j^{\prime}=1}^{J_{i}}\mathcal{P}_{\gamma}\left(\Vert{\tilde{\bf u}}_{\left(j^{\prime}\right)}\Vert_2,\sqrt{d_{ij^{\prime}}}\lambda_2\right) =\sum_{j^{\prime}=1}^{J_{i}}\mathcal{P}_{\gamma}\left(\Vert{\hat{\bf q}}_{i\left(j^{\prime}\right)}^{or}\Vert_2,\sqrt{d_{ij^{\prime}}}\lambda_2\right) = \frac{1}{2}\sum_{j^\prime}^{J_i} d_{ij^\prime}\gamma \lambda_2^2$.

Based on the definition of $\hat{{\bf q}}_{i}^{or}$, we have $\Vert\hat{{\bf S}}\hat{{\bf S}}^{\top}-\hat{{\bf s}}_{i}^{or}\left(\hat{{\bf s}}_{i}^{or}\right)^{\top}\Vert_{F}^{2}<\Vert\hat{{\bf S}}\hat{{\bf S}}^{\top}-\tilde{{\bf s}}\tilde{{\bf s}}^{\top}\Vert_{F}^{2}$
for $\tilde{{\bf u}}\neq\hat{{\bf q}}_{i}^{or}$, where $\hat{{\bf s}}_{i}^{or}=\left({\bf I}-\tilde{{\bf S}}_{i}\tilde{{\bf S}}_{i}^{\top}\right)\hat{{\bf q}}_{i}^{or}$
and $\tilde{{\bf s}}_{i}=\left({\bf I}-\tilde{{\bf S}}_{i}\tilde{{\bf S}}_{i}^{\top}\right)\tilde{{\bf u}}$.
This implies that $G\left(\hat{{\bf q}}_{i}^{or}\right)<G\left(\tilde{{\bf u}}\right)$.

Next, we will compare $G\left(\tilde{{\bf u}}\right)$ and $G\left({\bf u}\right)$. We have 
\begin{align}
\label{eq_gdiff2}
G\left(\tilde{{\bf u}}\right)-G\left({\bf u}\right) & =-\tilde{{\bf u}}^{\top}\hat{{\bf H}}\hat{{\bf S}}\hat{{\bf S}}^{\top}\hat{{\bf H}}\tilde{{\bf u}}+{\bf u}^{\top}\hat{{\bf H}}\hat{{\bf S}}\hat{{\bf S}}^{\top}\hat{{\bf H}}{\bf u} \nonumber \\
 & +\sum_{j=1}^{p}\mathcal{P}_{\gamma}\left(\vert\tilde{u}_{j}\vert,\lambda_{1}\right)-\sum_{j=1}^{p}\mathcal{P}_{\gamma}\left(\vert u_{j}\vert,\lambda_{1}\right)+ \nonumber\\
 & +\sum_{j^{\prime}=1}^{J_{i}}\mathcal{P}_{\gamma}\left(\Vert\tilde{{\bf u}}_{(j^{\prime})}\Vert_2,\sqrt{d_{ij^{\prime}}}\lambda_{2}\right)-\sum_{j^{\prime}=1}^{J_{i}}\mathcal{P}_{\gamma}\left(\Vert{\bf u}_{(j^{\prime})}\Vert_2,\sqrt{d_{ij^{\prime}}}\lambda_{2}\right)
\end{align}
where $\hat{{\bf H}}={\bf I}-\tilde{{\bf S}}_{i}\tilde{{\bf S}}_{i}^{\top}$. The bound of $-\tilde{{\bf u}}^{\top}\hat{{\bf H}}\hat{{\bf {\bf S}}}\hat{{\bf S}}^{\top}\hat{{\bf H}}\tilde{{\bf u}}+{\bf u}^{\top}\hat{{\bf H}}\hat{{\bf {\bf S}}}\hat{{\bf S}}^{\top}\hat{{\bf H}}{\bf u}$ is from the proof of Theorem 2 in \cite{liu::wang2025}, which has the following form,
\[
-\tilde{{\bf u}}^{\top}\hat{{\bf H}}\hat{{\bf {\bf S}}}\hat{{\bf S}}^{\top}\hat{{\bf H}}\tilde{{\bf u}}+{\bf u}^{\top}\hat{{\bf H}}\hat{{\bf {\bf S}}}\hat{{\bf S}}^{\top}\hat{{\bf H}}{\bf u}\leq O_{p}\left(\tau_{n,p,m}\right)\sum_{j\notin\mathcal{V}_{i}}\vert u_{j}\vert.
\]

For $\sum_{j=1}^{p}\mathcal{P}_{\gamma}\left(\vert\tilde{u}_{j}\vert,\lambda_{1}\right)-\sum_{j=1}^{p}\mathcal{P}_{\gamma}\left(\vert u_{j}\vert,\lambda_{1}\right)$ and $\sum_{j^{\prime}=1}^{J_{i}}\mathcal{P}_{\gamma}\left(\Vert\tilde{{\bf u}}_{(j^{\prime})}\Vert_2,\sqrt{d_{ij^{\prime}}}\lambda_{2}\right)-\sum_{j^{\prime}=1}^{J_{i}}\mathcal{P}_{\gamma}\left(\Vert{\bf u}_{(j^{\prime})}\Vert_2,\sqrt{d_{ij^{\prime}}}\lambda_{2}\right)$, the same arguemts can have as those in Part 1. Thus, we have have $G\left(\tilde{{\bf u}}\right)-G\left({\bf u}\right)<0$ for $\tilde{{\bf u}} \neq {\bf u}$.
Thus, we have shown that $G\left(\hat{{\bf q}}_{i}^{or}\right)<G\left(\tilde{{\bf u}}\right)<G\left({\bf u}\right)$,
which implies that $\hat{{\bf q}}_{i}^{or}$ is a local minimizer
of the objective function. This completes the proof.

\end{proof}

\section{Lemmas}
The proof of the theoretical properties of the proposed sparse group factor model is an extension of that in \cite{liu::wang2025} for the sparse factor model. We introduce some lemmas in \cite{liu::wang2025} first.  In Lemma \ref{lem_or}, we introduce the oracle property, which can be found in the proof of \cite{liu::wang2025}.

\begin{lemma}
Let ${\bf u}$ be a $p\times 1$ vector such that $\Vert{\bf u}\Vert_2=1$ 
and $\mathcal{A}$ be an index, ${\bf u}_{[\mathcal{A}]}^{*}={\bf u}_{[\mathcal{A}]}$, ${\bf u}_{[-\mathcal{A}]}^{*}={\bf 0}$ and $\vert \mathcal{A}\vert \asymp m$. Then, (the max element) \label{lemma_uE}
\[
\Vert{\bf u}^{*\top}\left(\hat{\bS}\hat{\bS}^{\top}-\bS\bS^{\top}\right)\Vert_{\max} =\begin{cases}
 O_p\left(\max\left(m^{2\delta-2}p^{2}n^{-1/2},m^{\delta}\right)\sqrt{\frac{\log p}{n}}\right), & \text{if } m= o(p)  \\
O_p\left( m^{\delta-1} p n^{-1/2}\right)  = O_p\left(p^\delta n^{-1/2}\right) & \text{if } m= O(p).  \\
\end{cases}
\]
\end{lemma}

First, we introduce the oracle estimator. The oracle estimator is defined as when the sparsity of ${\bf q}_i$ is known.

For ${\bf q}_1$, when the sparsity, $\mathcal{V}_1$, is known, the oracle estimator is defined as 
\begin{align*}
\hat{{\bf q}}_{1}^{or} & =\arg\min_{{\bf q}_{1}}\Vert\hat{{\bf S}}\hat{{\bf S}}^{\top}-{\bf q}_{1}{\bf q}_{1}^{\top}\Vert_{F}^{2}\\
 & \text{subject to }{\bf q}_{1[\mathcal{N}_{1}]}={\bf 0} \text{ and } \Vert {\bq_1} \Vert_2 =1 .
\end{align*}

Let ${\bf H}_i={\bf I}-{\bf S}_{i}{\bf S}_{i}^{\top}$,
\textbf{$\hat{{\bf H}}_i={\bf I}-\tilde{{\bf S}}_{i}\tilde{{\bf S}}_{i}^{\top}$}. When the sparsity of $\mathcal{V}_i$ is known, the oracle estimator is defined as, 
\begin{align}
\label{eq_ob_qi_or}
\hat{{\bf q}}_{i}^{or}=\arg\min_{{\bf q}_{i}}-{\bf q}_{i}^{\top}\hat{{\bf H}}_i\hat{{\bf S}}\hat{{\bf S}}_i^{\top}{\bf \hat{H}}_i{\bf q}_{i}\\
\text{subject to }{\bf q}_{i}^{\top}\hat{{\bf H}}_i{\bf \hat{H}}_i{\bf q}_{i}=1 \text{ and } \bq_{i[-\mathcal{V}_i]} ={\bf 0}. \nonumber
\end{align}

And denote $\hat{\bf q}_i$ as the estimator of ${\bf q}_i$. We have the following results about the oracle estimator following that in \cite{liu::wang2025}. Since the definition of the oracle estimator here is the same as that in \cite{liu::wang2025}.

\begin{lemma}
\label{lem_or}
   Under Conditions (C\ref{cond_alphamix})-(C\ref{cond_coherence}), we have 
   \begin{equation}
   \label{eq_or1}
        \Vert\hat{{\bf q}}_{1}^{or}-{\bf q}_{1}\Vert_{2}=O_{p}\left(\tau_{n,p,m}\right).
   \end{equation}
   
For $i\geq 2$, given that $\Vert \hat{{\bf q}}_{i-1} - {\bf q}_{i-1}\Vert_2= O_p(\tau_{n,p,m})$, then 
   \begin{equation}
   \label{eq_or2}
   \Vert\hat{{\bf q}}_{i}^{or}-{\bf q}_{i}\Vert_{2}=O_{p}\left(\tau_{n,p,m}\right).
   \end{equation}
\end{lemma}
\begin{proof}

$\Vert\hat{{\bf q}}_{1}^{or}-{\bf q}_{1}\Vert_{2}=O_{p}\left(\tau_{n,p,m}\right)$ comes directly from the result in \cite{liu::wang2025}. And from the part 1 proof of Theorem 1, we know that $\Vert \hat{{\bf q}}_1 - {\bf q}_1\Vert_2= O_p(\tau_{n,p,m})$ holds, then \eqref{eq_or2} follows that in \cite{liu::wang2025} to prove the results sequentially.
\end{proof}

\bibliographystyle{apalike}
\bibliography{reference}

\end{document}